\documentclass[lettersize,journal]{IEEEtran}

\usepackage{amsmath,amsfonts}
\usepackage{amsthm}
\usepackage{algorithmic}
\usepackage{algorithm}
\usepackage{array}
\usepackage[caption=false,font=footnotesize]{subfig}
\usepackage{textcomp}
\usepackage{stfloats}
\usepackage{url}
\usepackage{verbatim}
\usepackage{graphicx}
\usepackage{cite}
\usepackage{float}
\usepackage{booktabs}
\usepackage[table]{xcolor}
\usepackage{multirow}
\usepackage{threeparttable}
\usepackage{makecell}

\newtheorem{theorem}{Theorem}
\newtheorem{assumption}{Assumption}
\newtheorem{lemma}[theorem]{Lemma}

\begin{document}

\title{Windowed and Quantized Group-Based ADMM for Distributed Optimization in Heterogeneous Edge Networks}

\author{Gaiguo Wei,
        Qingying Zhang,
        Heqiang Wang,
        Yu Zhang,
        and Xiaoxiong Zhong%
\thanks{This work was supported by Pengcheng Laboratory Project under Grant PCL2025A13.}
\thanks{Gaiguo Wei and Qingying Zhang are with the Southern University of Science and Technology, Shenzhen 518055, China, and also with Pengcheng Laboratory, Shenzhen 518000, China (e-mail: 12447019@mail.sustech.edu.cn; 12547037@mail.sustech.edu.cn).}%
\thanks{Heqiang Wang and Xiaoxiong Zhong are with Pengcheng Laboratory, Shenzhen 518000, China (e-mail: wanghq02@pcl.ac.cn; xixzhong@gmail.com).}%
\thanks{Yu Zhang is with the Department of Computer Science and Engineering, Southern University of Science and Technology, Shenzhen 518055, China (e-mail: zhangy7@sustech.edu.cn).} 
\thanks{Corresponding authors: Yu Zhang and Xiaoxiong Zhong.}}
\maketitle

\begin{abstract}
Distributed optimization in edge networks is constrained by heterogeneous client computing capabilities and limited communication resources. We propose the Windowed and Quantized Group-Based Alternating Direction Method of Multipliers (WQ-GADMM) to coordinate group updates under limited activation capacity and reduce communication costs. Clients are grouped by estimated computation time. Each window activates a limited number of groups per round, and the cloud updates the global model after all groups have updated once. The method quantizes both downlink and uplink model exchanges to reduce communication costs and allows bounded model staleness and inexact proximal local updates. For smooth nonconvex objectives, we establish an average squared Karush-Kuhn-Tucker residual bound under the stated assumptions and parameter conditions. The bound consists of a term that decreases with the iteration count and a quantization-dependent error term. Experiments on MNIST and CIFAR-10 show that 12-bit communication reduces communication volume and simulated wall-clock time while maintaining test accuracy comparable to full precision. The 12-bit configuration also maintains complete group coverage and achieves shorter mean group inter-completion gaps than the evaluated baselines.
\end{abstract}


\section{Introduction}
\label{sec:introduction}

\IEEEPARstart{E}{dge} networks support collaborative model training across clients with locally collected data. Federated learning enables this collaboration through model exchanges while keeping raw data at the clients \cite{McMahan2017}. However, differences in computing capability, data volume, and communication conditions lead to unequal update completion times \cite{Kairouz2021}. These differences affect both local training and the collection of model updates. A cloud-edge-client architecture organizes collaboration through intermediate edge servers, each coordinating a group of clients while the cloud maintains the global model. Hierarchical federated learning uses this structure to coordinate local updates and reduce direct client-cloud communication \cite{Liu2020HFL}. Although this hierarchy organizes computation at the group level, efficient training still requires coordination of group participation, model exchanges, and global updates.

Collaborative training in this architecture can be formulated as a distributed optimization problem. Each group objective combines the losses of its member clients, and consensus constraints link the group models to a global model. The alternating direction method of multipliers (ADMM) decomposes this problem into local subproblems and coordinates their solutions through global and dual updates \cite{Boyd2011}. Its distributed execution depends on the resources available for group updates. We consider a setting in which only a limited number of edge groups can be activated per physical communication round. When the group count exceeds this limit, collecting a contribution from every group requires multiple rounds. The update procedure must therefore specify which groups participate in each round and when the cloud updates the global model.

Beyond coordinating group participation, distributed ADMM must also address communication and computation costs. Each group exchanges model information with the cloud and solves a local subproblem. Quantization reduces the transmitted model payload, while approximate local solutions can reduce the computational effort required for each update. Both introduce approximation errors into the optimization process. If groups use cached global models, their local updates also depend on information from earlier iterations. The analysis must therefore account for model staleness together with quantization and local solution errors.

Existing methods address these issues through different mechanisms. Parallel Group-Based ADMM (P-GADMM) forms computation-aware groups and uses bounded asynchronous coordination to accommodate unequal group completion times \cite{Wei2026PGADMM}. Quantized Group ADMM (Q-GADMM) combines group coordination with quantized model exchanges to reduce communication in decentralized learning \cite{Elgabli2020QGADMM}. Inexact ADMM allows approximate local solutions and has been applied to federated learning \cite{Zhou_InexactADMM}. These methods address group coordination, communication reduction, and local computation under different update rules. Our work considers these mechanisms together under limited activation capacity, with one update from every group before each global update.

To address these challenges, we propose Windowed and Quantized Group-Based ADMM (WQ-GADMM). Clients are assigned to fixed edge groups according to their estimated computation times. Each activation window spans the physical rounds needed for all groups to complete one update. In each round, the scheduling rule selects groups that have not yet updated, subject to the activation limit. The global model remains fixed throughout the window, and the cloud performs the global and dual updates after all groups finish. Each window therefore corresponds to one logical ADMM iteration. Bidirectional quantization reduces the payload of downlink references and uplink models. Local updates use cached references with bounded model staleness and approximate solutions to proximal subproblems. Fig.~\ref{fig:design_overview} summarizes the design.

\begin{figure}[t]
\centering
\includegraphics[width=\linewidth]{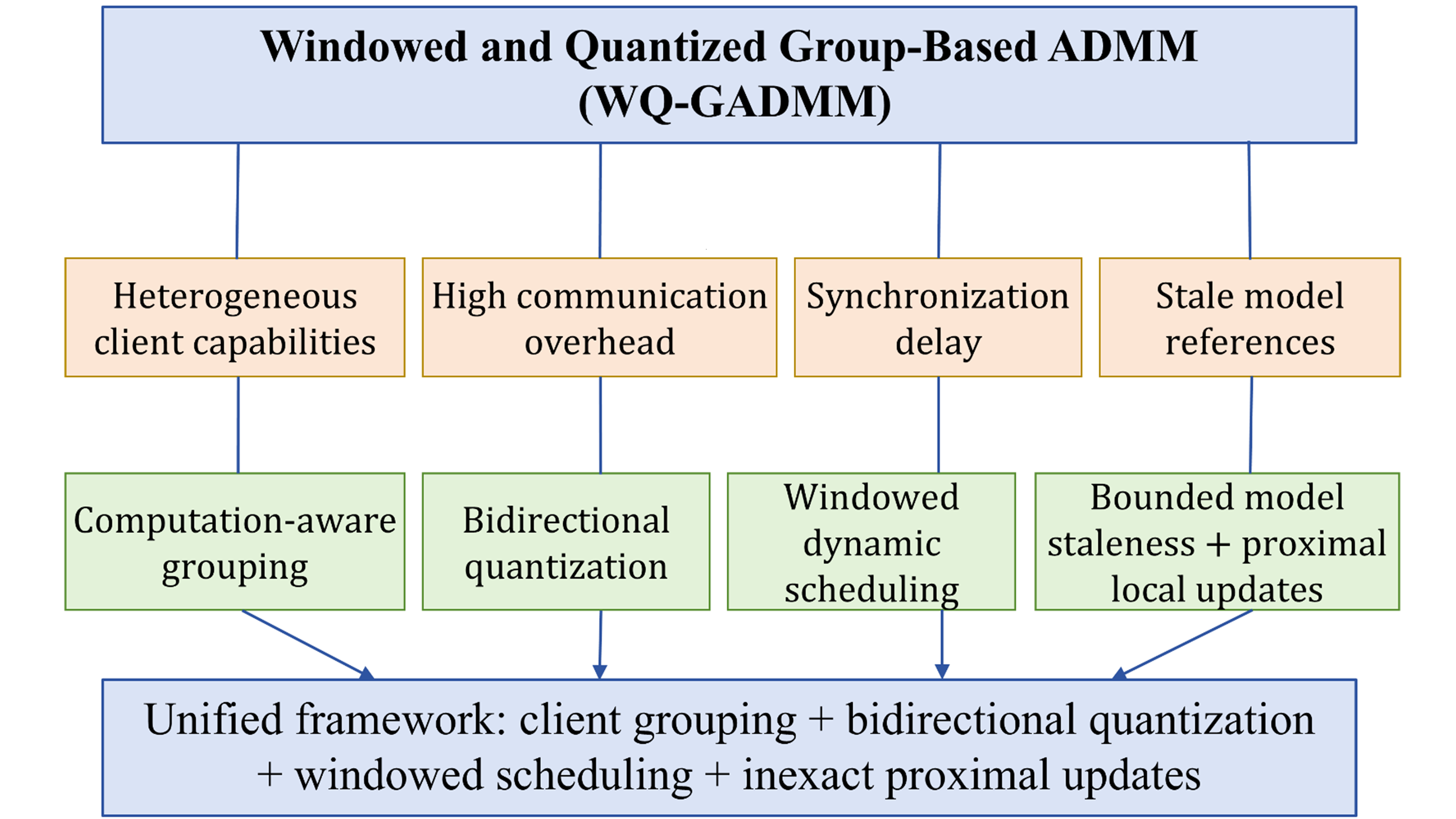}
\caption{Main challenges in heterogeneous edge networks and the corresponding components of WQ-GADMM.}
\label{fig:design_overview}
\end{figure}

The main contributions of this work are as follows.
\begin{enumerate}

\item We propose WQ-GADMM to coordinate group updates and reduce communication costs in heterogeneous edge networks. Clients are divided into fixed groups based on estimated computation times. Within each activation window, a limited number of groups update per communication round, and the cloud updates the global model after every group has updated once. The method uses bidirectional quantization and allows bounded model staleness and inexact local updates.

\item We establish an average squared Karush-Kuhn-Tucker (KKT) residual bound for smooth nonconvex objectives under the stated assumptions and parameter conditions. The bound accounts for model staleness and inexact local updates, and consists of an $O(1/K)$ term and a quantization-dependent error term. It guarantees a bounded average residual in the presence of these errors.

\item We evaluate WQ-GADMM on a smooth nonconvex test problem and the MNIST and CIFAR-10 datasets, examining residual behaviour, learning performance, group participation, and resource efficiency. Under the same training workload, the 12-bit configuration (Q12) reduces communication volume by approximately 62.5\% compared with the 32-bit floating-point configuration (FP32), while achieving comparable test accuracy and lower simulated wall-clock time. Q12 also maintains complete group coverage and achieves shorter group inter-completion gaps than the baselines.

\end{enumerate}

The remainder of this paper is organized as follows. Section~\ref{sec:related_work} reviews related work, and Section~\ref{sec:system_model} presents the system model and problem formulation. Section~\ref{sec:proposed_algorithm} describes WQ-GADMM, and Section~\ref{sec:convergence} provides its convergence analysis. Section~\ref{sec:experiments} reports the experimental results, followed by the conclusion in Section~\ref{sec:conclusion}.

\section{Related Work}
\label{sec:related_work}

\subsection{Distributed Optimization in Edge Networks}

Distributed optimization coordinates local computation and information exchange across participants to minimize a global objective. Distributed subgradient and dual averaging methods provide convergence guarantees for this process \cite{Nedic2009,Duchi2012}, while communication-efficient statistical optimization studies its computational and communication costs \cite{Zhang2013}. In heterogeneous edge networks, local data volumes affect computational workloads, and available resources influence update completion times. FedProx addresses data and system heterogeneity through proximal local objectives and flexible amounts of local computation \cite{Li2020FedProx}. These studies provide a basis for coordinating local optimization under heterogeneous data and resource conditions.

Client organization and selection offer additional ways to address heterogeneity. Hierarchical federated learning coordinates updates across network levels \cite{Abad2020}, while adaptive clustering groups clients according to their data characteristics \cite{Tian2022}. Client selection can also account for participation fairness \cite{Shi2023} or training efficiency. For example, Oort combines data utility with execution speed when selecting clients \cite{Lai2021Oort}. These approaches address heterogeneity through network hierarchy, client grouping, and participant selection.

\subsection{ADMM and Inexact Local Optimization}

ADMM has been applied to distributed learning and optimization in edge networks \cite{He2025}, and its convergence has been studied beyond convex objectives. Hong et al. analyze convergence for nonconvex consensus and sharing problems under suitable assumptions and penalty parameters \cite{Hong2016NonconvexADMM}. Wang et al. establish sufficient conditions for convergence in a class of nonconvex, nonsmooth problems with coupled linear equality constraints \cite{Wang2019NonconvexADMM}. In distributed implementations, solving local subproblems exactly can be computationally expensive, motivating inexact updates. Proximal point and operator splitting methods provide a theoretical foundation for controlling the resulting approximation errors \cite{Eckstein1992DRS}. Chang et al. use proximal gradient updates to reduce local solution costs \cite{Chang2015InexactCADMM}, while stochastic ADMM uses sampled data for local updates \cite{Ouyang_StochasticADMM}. Inexact ADMM and three-operator ADMM have also been developed for federated learning \cite{Zhou_InexactADMM,Kant_ThreeOpADMM}.

In addition to local computation, synchronization affects the execution time of distributed ADMM. Zhang and Kwok combine a partial barrier with bounded delay to coordinate asynchronous updates \cite{Zhang2014Async}. Chang et al. analyze asynchronous distributed ADMM for possibly nonconvex objectives and derive convergence conditions that depend on the maximum network delay \cite{Chang2016Async}. These methods allow global updates using fresh results from a subset of participants while limiting update delays.

\subsection{Quantized and Group-Based ADMM}

Quantized consensus methods examine how finite communication precision affects agreement among participants \cite{Kashyap2007QuantizedConsensus}. Quantized consensus ADMM incorporates quantized exchanges into distributed optimization and analyzes the resulting errors \cite{Zhu2016QC_ADMM}. For asynchronous ADMM, Shrestha combines bidirectional compression with error feedback to reduce communication costs \cite{Shrestha2025AsyncQADMM}. Related approaches address communication costs in distributed stochastic gradient training. Quantized stochastic gradient descent (QSGD) characterizes the tradeoff between transmitted bits and quantization variance \cite{Alistarh2017QSGD}, while DoubleSqueeze applies compression with error compensation at both workers and the parameter server \cite{Tang2019DoubleSqueeze}. These studies explore communication reduction through quantization and compression across different optimization methods and update schemes.

Group-based methods organize local updates and model exchanges, providing a complementary approach to quantization. Group-based ADMM uses a group structure for distributed classification \cite{Wang2017GADMM}. Group ADMM (GADMM) and its extensions explore group coordination, layer-wise updates, and different communication patterns \cite{Elgabli2020GADMM,Elgabli2020LFGADMM,Huang2021GRADMM,BenIssaid2022}. Q-GADMM combines group coordination with quantized exchanges for decentralized learning \cite{Elgabli2020QGADMM}, while P-GADMM uses computation-aware groups and bounded asynchronous cloud updates \cite{Wei2026PGADMM}. WQ-GADMM considers a hierarchical edge network in which only a limited number of groups can be activated per physical round. It uses fixed groups and activation windows, with each group completing one update before the cloud updates the global model. Both downlink references and uplink models are quantized to reduce the transmitted model payload.

\section{System Model and Problem Formulation}
\label{sec:system_model}

We consider a heterogeneous edge network with clients, edge servers, and a cloud server. Clients retain their data locally and collaborate in model training. Edge servers coordinate client groups, while the cloud maintains the global model and coordinates group updates. Clients differ in their data volumes, computing resources, and communication capabilities. 

\subsection{Client Model and Fixed Group Partition}

Let $\mathcal C=\{1,2,\ldots,M\}$ denote the client set. Client $i$ holds a local dataset $D_i$ containing $n_i=|D_i|$ samples. Its empirical loss is
\begin{equation}
F_i(w)
=\frac{1}{n_i}\sum_{\xi\in D_i}\ell_i(w;\xi),
\label{eq:local_risk}
\end{equation}
where $w\in\mathbb R^d$ is the model parameter vector and $\ell_i(w;\xi)$ is the loss on sample $\xi$.

Before training, the clients are partitioned into $G$ nonempty, disjoint groups $\{\mathcal C_g\}_{g=1}^{G}$ satisfying
\begin{equation}
\begin{aligned}
&\mathcal C_g\neq\emptyset,
\qquad
\bigcup_{g=1}^{G}\mathcal C_g=\mathcal C,\\
&\mathcal C_g\cap\mathcal C_{g'}=\emptyset,
\qquad g\neq g'.
\end{aligned}
\label{eq:fixed_group_partition}
\end{equation}
The grouping procedure in Section~\ref{subsec:grouping_strategy} uses estimated client computation times. Group membership remains fixed throughout training, while WQ-GADMM determines which groups are activated in each physical round.

Let $n_g=\sum_{i\in\mathcal C_g}n_i$ and $n=\sum_{g=1}^{G}n_g$ denote the group and total sample counts, respectively. The objective of group $g$ is
\begin{equation}
\phi_g(w_g)
=\sum_{i\in\mathcal C_g}\frac{n_i}{n}F_i(w_g).
\label{eq:group_objective}
\end{equation}
Equivalently, $\phi_g(w_g)=(n_g/n)\bar F_g(w_g)$, where $\bar F_g(w_g)=\sum_{i\in\mathcal C_g}(n_i/n_g)F_i(w_g)$ is the group's sample-averaged loss. This definition preserves the sample weights in the global objective:
\begin{equation}
\sum_{g=1}^{G}\phi_g(w)
=\sum_{i=1}^{M}\frac{n_i}{n}F_i(w).
\label{eq:global_objective_equivalence}
\end{equation}

\subsection{Group Consensus Formulation}

Each edge group maintains a model $w_g\in\mathbb R^d$, and the cloud maintains a global model $w\in\mathbb R^d$. Collaborative training is formulated as
\begin{equation}
\begin{aligned}
\underset{\{w_g\}_{g=1}^{G},\,w}{\operatorname{minimize}}
\quad&
\sum_{g=1}^{G}\phi_g(w_g)\\
\operatorname{subject\ to}
\quad&
w_g=w,\qquad g=1,\ldots,G.
\end{aligned}
\label{eq:consensus_problem}
\end{equation}
The consensus constraints require all group models to agree with the global model.

Let $\lambda_g$ be the dual variable associated with $w_g=w$. For a fixed ADMM penalty parameter $\rho>0$, define the scaled dual variable $u_g=\lambda_g/\rho$, so that
\begin{equation}
\lambda_g=\rho u_g.
\label{eq:scaled_dual_definition}
\end{equation}

For differentiable group objectives, a first-order stationary point $(w^*,\{w_g^*\}_{g=1}^{G})$ with associated scaled multipliers $\{u_g^*\}_{g=1}^{G}$ satisfies the KKT conditions
\begin{equation}
\nabla\phi_g(w_g^*)+\rho u_g^*=0,
\qquad g=1,\ldots,G,
\label{eq:kkt_stationarity}
\end{equation}
\begin{equation}
w_g^*=w^*,
\qquad g=1,\ldots,G,
\label{eq:kkt_consensus}
\end{equation}
and
\begin{equation}
\sum_{g=1}^{G}u_g^*=0.
\label{eq:kkt_dual_balance}
\end{equation}
These conditions provide the basis for the KKT residual used in Section~\ref{sec:convergence}.

\section{Windowed and Quantized Group-Based ADMM}
\label{sec:proposed_algorithm}

WQ-GADMM addresses the consensus problem in \eqref{eq:consensus_problem} through fixed client groups, windowed scheduling, and bidirectional quantization. Each logical iteration spans an activation window, with at most $M_a$ groups activated per physical communication round. Every group completes one update before the cloud updates the global model and dual variables.

Figure~\ref{fig:workflow} illustrates the procedure. An activated group receives a quantized reference, approximately solves a proximal subproblem, and uploads a quantized model. The global model and dual variables remain fixed within each window. We describe group construction, scheduling, local computation, and cloud updates below.

\begin{figure*}[t]
\centering
\includegraphics[width=\textwidth]{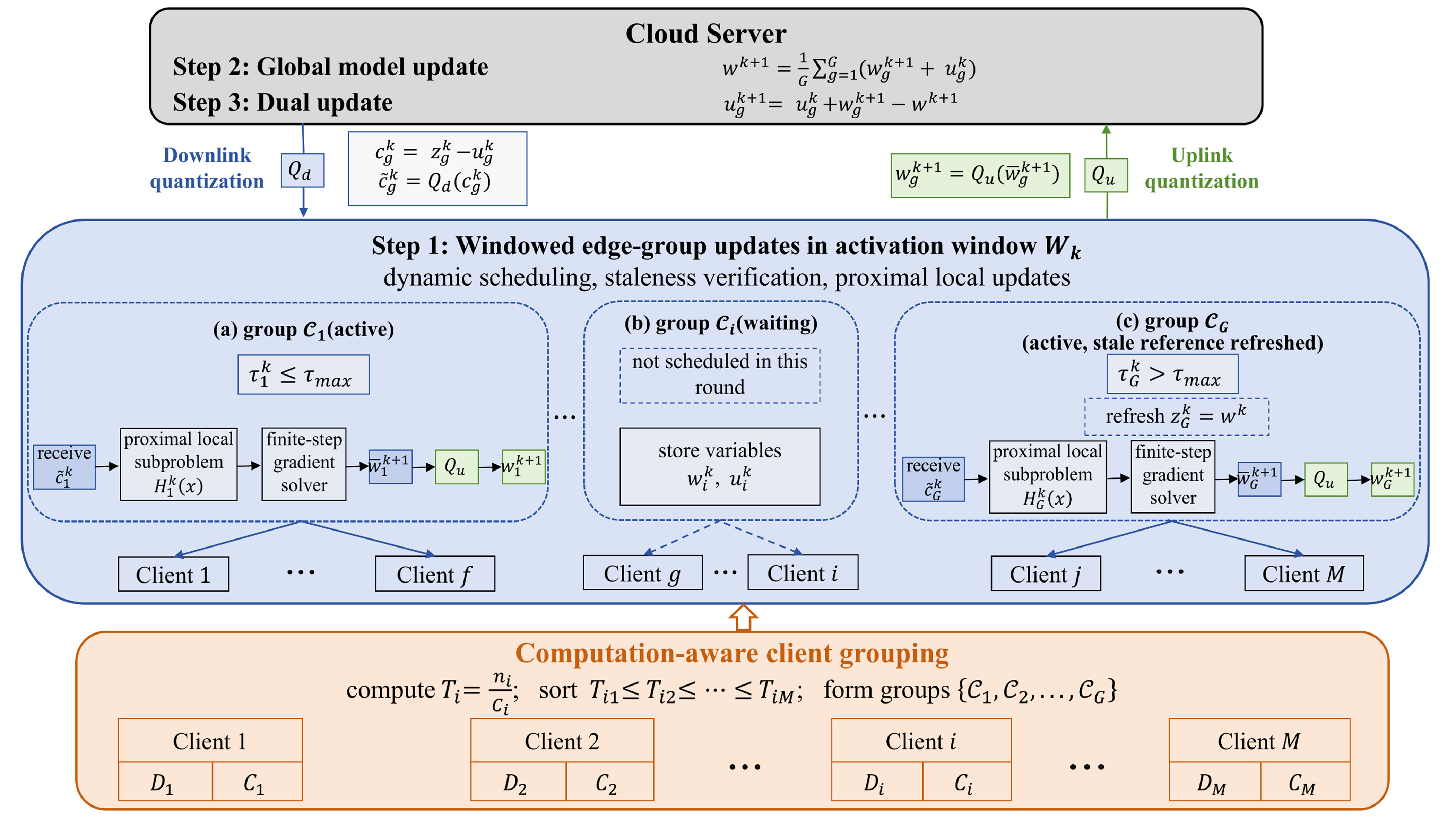}
\caption{Workflow of WQ-GADMM. Groups perform local updates over multiple physical rounds using quantized downlink references and uplink models. The cloud updates the global model and dual variables after every group has contributed once.}
\label{fig:workflow}
\end{figure*}

\subsection{Computation-Aware Grouping}
\label{subsec:grouping_strategy}

Before training, clients are partitioned into $G$ fixed groups according to their estimated computation times, where $1\le G\le M$. Let $C_i>0$ denote the effective sample processing rate of client $i$, measured in samples per second. The estimated computation time used for grouping is
\begin{equation}
T_i=\frac{n_i}{C_i}.
\label{eq:estimated_client_latency}
\end{equation}
Let $(i_1,\ldots,i_M)$ be a permutation satisfying
\begin{equation}
T_{i_1}\le T_{i_2}\le\cdots\le T_{i_M}.
\label{eq:latency_sorted_clients}
\end{equation}
Define $b_g=\lfloor gM/G\rfloor$ for $g=0,\ldots,G$. Group $g$ is
\begin{equation}
\mathcal C_g
=\{i_j:b_{g-1}<j\le b_g\},
\qquad g=1,\ldots,G.
\label{eq:computation_aware_groups}
\end{equation}
This construction groups clients with similar estimated computation times. Group membership remains fixed throughout training.

\subsection{Windowed Group Scheduling}

Let $k$ denote the logical iteration index and $r$ the physical communication round index. The set $\Omega_k$ contains groups that have completed their updates in the current window. At the start of each window, $\Omega_k=\emptyset$ and all waiting counters are reset to zero. The eligible groups are
\begin{equation}
\mathcal E_k(r)
=\{1,\ldots,G\}\setminus\Omega_k.
\label{eq:eligible_set}
\end{equation}
At most $M_a\in\{1,\ldots,G\}$ groups can be activated in each round.

For each eligible group, $s_g^r$ counts the preceding rounds in the current window in which it was not selected. Given an integer threshold $T_{\mathrm{act}}\ge1$, define
\begin{align}
\widehat U_k(r)
&=\{g\in\mathcal E_k(r):s_g^r\ge T_{\mathrm{act}}-1\},
\label{eq:forced_candidate_set}\\
\widehat M_r
&=\min\{M_a,|\widehat U_k(r)|\},
\label{eq:forced_activation_number}\\
U_k(r)
&=\operatorname{Top}_{\widehat M_r}
\bigl(\widehat U_k(r),s_g^r\bigr).
\label{eq:forced_set}
\end{align}
Here, $\operatorname{Top}_{\widehat M_r}(\widehat U_k(r),s_g^r)$ selects the $\min\{\widehat M_r,|\widehat U_k(r)|\}$ groups in $\widehat U_k(r)$ with the largest waiting counts.

All eligible groups have the same waiting count. Their counters are reset to zero at the start of each window and increase by one in each round in which they are not selected. They therefore reach the threshold in the same round, with ties resolved in ascending order of group index.

The remaining candidates and available slots are
\begin{align}
R_k(r)
&=\mathcal E_k(r)\setminus U_k(r),
\label{eq:remaining_set}\\
M_r
&=\min\{M_a-|U_k(r)|,\;|R_k(r)|\}.
\label{eq:remaining_capacity}
\end{align}
For each remaining candidate, compute
\begin{equation}
S_g^r
=\omega_1
\frac{\|w_g^k-w^k\|^2}
{\max\{\|w^k\|^2,\epsilon_s\}}
+\omega_2s_g^r,
\label{eq:scheduling_score}
\end{equation}
where $\omega_1,\omega_2\ge0$ are fixed weights for model disagreement and waiting count, respectively, and $\epsilon_s>0$ is a constant that prevents division by zero. Since all eligible groups have the same waiting count, the waiting term does not affect their ranking. Before the threshold is reached, groups are therefore ranked by their normalized model disagreement when $\omega_1>0$.

The remaining slots are filled by
\begin{equation}
B_k(r)
=\operatorname{Top}_{M_r}\bigl(R_k(r),S_g^r\bigr),
\label{eq:score_selected_set}
\end{equation}
giving the active set
\begin{equation}
A_k(r)=U_k(r)\cup B_k(r),
\qquad |A_k(r)|\le M_a.
\label{eq:active_set}
\end{equation}

Groups in $A_k(r)$ execute their updates concurrently. After all active groups complete, the cloud records
\begin{equation}
\Omega_k\leftarrow\Omega_k\cup A_k(r).
\label{eq:completed_set_update}
\end{equation}
The waiting counters are updated using the eligible set defined at the start of round $r$:
\begin{equation}
s_g^{r+1}=
\begin{cases}
s_g^r+1,
&g\in\mathcal E_k(r)\setminus A_k(r),\\
0,
&\text{otherwise}.
\end{cases}
\label{eq:waiting_time_update}
\end{equation}

The activation window is
\begin{equation}
\mathcal W_k
=\{r_k,r_k+1,\ldots,r_{k+1}-1\},
\label{eq:activation_window}
\end{equation}
where $r_k$ is its first round and $r_{k+1}$ is the first round of the next window. The window ends after all groups have completed their updates:
\begin{equation}
\bigcup_{r\in\mathcal W_k}A_k(r)=\{1,\ldots,G\}.
\label{eq:window_coverage}
\end{equation}
Completed groups are excluded from subsequent rounds, ensuring exactly one update per group within each window. Each round activates as many eligible groups as the capacity allows. Therefore, each window contains $\lceil G/M_a\rceil$ physical rounds, provided all activated updates finish.

\subsection{Quantized Local Updates with Bounded Staleness}

The cloud initializes $z_g^0=w^0$ and $d_g^0=0$ for all groups. At the start of iteration $k\ge1$, it retains the cached model and its index from the previous iteration. The cached global model is
\begin{equation}
z_g^k=w^{d_g^k},
\label{eq:cached_reference}
\end{equation}
where $d_g^k$ denotes the logical iteration index of the cached global model. When group $g$ is activated, the cloud computes the model staleness as
\begin{equation}
\tau_g^k=k-d_g^k.
\label{eq:model_staleness}
\end{equation}
For a fixed nonnegative integer $\tau_{\max}$, the cache is refreshed whenever $\tau_g^k>\tau_{\max}$:
\begin{equation}
z_g^k\leftarrow w^k,
\qquad
d_g^k\leftarrow k,
\qquad
\tau_g^k\leftarrow0.
\label{eq:reference_refresh}
\end{equation}
After the refresh check, the model staleness satisfies
\begin{equation}
0\le\tau_g^k\le\tau_{\max}.
\label{eq:bounded_staleness}
\end{equation}

The cloud constructs the downlink reference
\begin{equation}
c_g^k=z_g^k-u_g^k
\label{eq:downlink_reference}
\end{equation}
and sends its quantized version to group $g$:
\begin{equation}
\widetilde c_g^k
=Q_d(c_g^k)
=c_g^k+\varepsilon_{d,g}^k,
\label{eq:downlink_quantization}
\end{equation}
where $Q_d$ is the downlink quantizer and $\varepsilon_{d,g}^k$ is the downlink quantization error.

Using the received reference $\widetilde c_g^k$, group $g$ approximately minimizes the local objective
\begin{equation}
H_g^k(x)
=\phi_g(x)
+\frac{\rho}{2}\|x-\widetilde c_g^k\|^2
+\frac{\eta}{2}\|x-w_g^k\|^2,
\label{eq:local_subproblem}
\end{equation}
where $\eta>0$ is the proximal parameter. The gradient of the group objective is
\begin{equation}
\nabla\phi_g(x)
=\sum_{i\in\mathcal C_g}\frac{n_i}{n}\nabla F_i(x).
\label{eq:distributed_group_gradient}
\end{equation}
The gradient of the local objective is therefore
\begin{equation}
\nabla H_g^k(x)
=\nabla\phi_g(x)
+\rho(x-\widetilde c_g^k)
+\eta(x-w_g^k).
\label{eq:local_objective_gradient}
\end{equation}

Starting from
\begin{equation}
x_g^{k,0}=w_g^k,
\label{eq:local_initialization}
\end{equation}
the group performs gradient updates
\begin{equation}
x_g^{k,m+1}
=x_g^{k,m}-\zeta_g\nabla H_g^k(x_g^{k,m}),
\qquad m=0,1,\ldots.
\label{eq:local_gd}
\end{equation}
Let $L_g$ be a Lipschitz constant of $\nabla\phi_g$. Choose $\rho,\eta>0$ with $\rho+\eta>L_g$, and set
\begin{equation}
0<\zeta_g\le\frac{1}{L_g+\rho+\eta}.
\label{eq:local_stepsize}
\end{equation}
For a tolerance $\theta_g>0$, let $M_g(k)$ be the first index $m\ge0$ satisfying
\[
\|\nabla H_g^k(x_g^{k,m})\|
\le\theta_g\|x_g^{k,m}-w_g^k\|.
\]
The unquantized output is
\begin{equation}
\bar w_g^{k+1}=x_g^{k,M_g(k)}.
\label{eq:unquantized_local_model}
\end{equation}
Define its local optimality residual as
\begin{equation}
e_g^{k+1}=\nabla H_g^k(\bar w_g^{k+1}).
\label{eq:local_optimality_residual}
\end{equation}
The stopping criterion ensures
\begin{equation}
\|e_g^{k+1}\|
\le\theta_g\|\bar w_g^{k+1}-w_g^k\|.
\label{eq:local_residual_bound}
\end{equation}

The group uploads the quantized model
\begin{equation}
w_g^{k+1}
=Q_u(\bar w_g^{k+1})
=\bar w_g^{k+1}+\varepsilon_{u,g}^{k+1},
\label{eq:uplink_quantization}
\end{equation}
where $Q_u$ is the uplink quantizer and $\varepsilon_{u,g}^{k+1}$ is the quantization error. 

\subsection{Cloud Update and Overall Algorithm}

After all groups complete their updates within $\mathcal W_k$, the cloud computes
\begin{equation}
w^{k+1}
=\frac1G\sum_{g=1}^{G}(w_g^{k+1}+u_g^k)
\label{eq:global_update}
\end{equation}
and updates the scaled dual variables:
\begin{equation}
u_g^{k+1}
=u_g^k+w_g^{k+1}-w^{k+1},
\qquad g=1,\ldots,G.
\label{eq:dual_update_scaled}
\end{equation}
The group consensus residual is
\begin{equation}
r_g^{k+1}=w_g^{k+1}-w^{k+1}.
\label{eq:consensus_residual}
\end{equation}
With the fixed penalty parameter $\rho$ and $\lambda_g^k=\rho u_g^k$ from \eqref{eq:scaled_dual_definition}, the dual increments satisfy
\begin{equation}
\begin{aligned}
u_g^{k+1}-u_g^k&=r_g^{k+1},\\
\lambda_g^{k+1}-\lambda_g^k&=\rho r_g^{k+1}.
\end{aligned}
\label{eq:dual_residual_relation}
\end{equation}

\begin{algorithm}[!t]
\caption{Windowed and Quantized Group-Based ADMM (WQ-GADMM)}
\label{alg:wq_gadmm}
\begin{algorithmic}[1]

\REQUIRE Client profiles $\{(C_i,n_i)\}_{i=1}^{M}$; number of groups $G$; initial model $w^0$; iteration limit $K$.
\REQUIRE Parameters $\rho,\eta,M_a,T_{\mathrm{act}},\tau_{\max}$; scheduling weights $\omega_1,\omega_2$; score stabilizer $\epsilon_s$; quantizers $Q_d,Q_u$; local stepsizes $\{\zeta_g\}_{g=1}^{G}$ and tolerances $\{\theta_g\}_{g=1}^{G}$.

\STATE Form fixed client groups based on estimated computation times, following Section~\ref{subsec:grouping_strategy}.
\STATE Initialize $w_g^0=z_g^0=w^0$, $u_g^0=0$, and $d_g^0=0$ for all $g$; set $r=0$.

\FOR{$k=0,1,\ldots,K-1$}
    \IF{$k>0$}
        \STATE Retain the cached models and indices: $(z_g^k,d_g^k)\leftarrow(z_g^{k-1},d_g^{k-1})$ for all $g$.
    \ENDIF
    \STATE Reset the completed-group set $\Omega_k\leftarrow\emptyset$ and waiting counters $s_g^r\leftarrow0$.

    \WHILE{$\Omega_k\neq\{1,\ldots,G\}$}
        \STATE Select active groups $A_k(r)$ using \eqref{eq:eligible_set}-\eqref{eq:active_set}.

        \FOR{each $g\in A_k(r)$ in parallel}
            \STATE Compute the cache age $\tau_g^k=k-d_g^k$.
            \IF{$\tau_g^k>\tau_{\max}$}
                \STATE Refresh the cached model and index: $z_g^k\leftarrow w^k$, $d_g^k\leftarrow k$, and $\tau_g^k\leftarrow0$.
            \ENDIF
            \STATE Transmit the quantized reference $\widetilde c_g^k=Q_d(z_g^k-u_g^k)$ to group $g$.
            \STATE Initialize the local iterate $x_g^{k,0}=w_g^k$.
            \STATE Apply the gradient update in \eqref{eq:local_gd} until the stopping criterion defining $M_g(k)$ is satisfied.
            \STATE Set $\bar w_g^{k+1}=x_g^{k,M_g(k)}$ and transmit $w_g^{k+1}=Q_u(\bar w_g^{k+1})$ to the cloud.
        \ENDFOR

        \STATE Update the waiting counters using \eqref{eq:waiting_time_update}.
        \STATE Add the completed groups to $\Omega_k\leftarrow\Omega_k\cup A_k(r)$ and increment $r\leftarrow r+1$.
    \ENDWHILE

    \STATE Update the global model using \eqref{eq:global_update} and the dual variables using \eqref{eq:dual_update_scaled}.
\ENDFOR

\STATE \textbf{return} $w^K$.
\end{algorithmic}
\end{algorithm}

The complete procedure of WQ-GADMM is summarized in Algorithm~\ref{alg:wq_gadmm}.

\section{Convergence Analysis}
\label{sec:convergence}

In this section, we establish an average squared KKT residual bound for WQ-GADMM with smooth nonconvex objectives. The bound consists of an $O(1/K)$ term and a quantization-dependent error term.

\subsection{Assumptions and Notation}

We consider Algorithm~\ref{alg:wq_gadmm} with its stated initialization, window updates, and cache refresh rule. The local solver follows \eqref{eq:local_gd} and satisfies the stopping criterion in \eqref{eq:local_residual_bound}.

\begin{assumption}[Smooth objectives]
\label{assu:nonconvex_smooth}
Each group objective $\phi_g:\mathbb R^d\to\mathbb R$ is differentiable and satisfies
\[
\|\nabla\phi_g(x)-\nabla\phi_g(y)\|
\le L_g\|x-y\|,
\qquad x,y\in\mathbb R^d.
\]
The objectives may be nonconvex, and their sum is bounded below:
\[
F(w):=\sum_{g=1}^G\phi_g(w)\ge F_{\inf}>-\infty,
\qquad w\in\mathbb R^d.
\]
\end{assumption}

\begin{assumption}[Quantization errors]
\label{assu:bounded_quantization}
For every group and iteration, the quantization errors satisfy
\begin{equation}
\mathbb E\!\left[
\|\varepsilon_{d,g}^k\|^2
+\|\varepsilon_{u,g}^{k+1}\|^2
\,\middle|\,\mathcal F_k
\right]
\le \sigma_{q,g}^2<\infty,
\label{eq:bounded_quant_error}
\end{equation}
where $\mathcal F_k$ contains the history before the quantized updates in iteration $k$. The constants $\sigma_{q,g}^2$ do not depend on $k$.
\end{assumption}

\begin{assumption}[Finite second moments]
\label{assu:finite_moments}
At every finite iteration $k$, the model and dual variables satisfy
\[
\mathbb E\!\left[
\|w^k\|^2+
\sum_{g=1}^G
\left(\|w_g^k\|^2+\|u_g^k\|^2\right)
\right]<\infty.
\]
\end{assumption}

Throughout this section, $\rho,\eta,\theta_g>0$, the cache threshold $\tau_{\max}$ is a finite nonnegative integer, and $0<\zeta_g\le(L_g+\rho+\eta)^{-1}$. Define
\begin{equation}
L=\max_g L_g,\qquad
\theta=\max_g\theta_g,\qquad
\sigma_q^2=\sum_{g=1}^G\sigma_{q,g}^2.
\label{eq:aggregate_quantization_bound}
\end{equation}

Let $x^k=\operatorname{col}(w_1^k,\ldots,w_G^k)$ and $\lambda^k=\rho\operatorname{col}(u_1^k,\ldots,u_G^k)$ denote the stacked group models and unscaled dual variables. The function $f(x)=\sum_g\phi_g(x_g)$ has an $L$-Lipschitz continuous gradient. For $k\ge1$, define
\begin{equation}
\begin{aligned}
s^k&=x^k-x^{k-1},
&
v^k&=w^k-w^{k-1},\\
r^k&=x^k-\mathbf1\otimes w^k,
&
S_k&=\mathbb E\|s^k\|^2.
\end{aligned}
\label{eq:nc_increments}
\end{equation}
Here, $\mathbf1\in\mathbb R^G$ is the all-ones vector, and $\otimes$ denotes the Kronecker product. The vectors $s^k$ and $v^k$ represent group and global model changes, respectively, while $r^k=\operatorname{col}(r_1^k,\ldots,r_G^k)$ is the stacked consensus residual.

\subsection{Auxiliary Lemmas}

\begin{lemma}[Termination of the local solver]
\label{lemma:local_termination}
Under Assumption~\ref{assu:nonconvex_smooth}, if $\rho+\eta>L$, every local gradient loop in Algorithm~\ref{alg:wq_gadmm} terminates after finitely many steps.
\end{lemma}

\begin{IEEEproof}
Fix a group and iteration, with the received reference held constant. Write $H=H_g^k$, $x^m=x_g^{k,m}$, and $x^0=w_g^k$. Smoothness gives
\[
\phi_g(y)\ge
\phi_g(x)+\langle\nabla\phi_g(x),y-x\rangle
-\frac{L_g}{2}\|y-x\|^2.
\]
Adding the quadratic terms in \eqref{eq:local_subproblem} shows that $H$ is $\mu_H$-strongly convex, where $\mu_H=\rho+\eta-L_g>0$. Its gradient is $L_H$-Lipschitz continuous, with $L_H=L_g+\rho+\eta$, and it has a unique minimizer $x^*$.

The prescribed gradient step satisfies
\[
H(x^{m+1})
\le H(x^m)-\frac{\zeta_g}{2}\|\nabla H(x^m)\|^2.
\]
For the sequence continued without stopping, summation yields
\[
\sum_{m=0}^{\infty}\|\nabla H(x^m)\|^2
\le
\frac{2[H(x^0)-H(x^*)]}{\zeta_g}<\infty.
\]
Thus $\nabla H(x^m)\to0$, and strong convexity implies
\[
\|x^m-x^*\|
\le\frac{\|\nabla H(x^m)\|}{\mu_H}\to0.
\]
If $x^0=x^*$, the stopping criterion holds immediately. Otherwise, set $D=\|x^*-x^0\|>0$. For all sufficiently large finite $m$,
\[
\|x^m-x^0\|\ge D/2,
\qquad
\|\nabla H(x^m)\|\le\theta_gD/2.
\]
Hence $\|\nabla H(x^m)\|\le\theta_g\|x^m-x^0\|$, which proves finite termination.
\end{IEEEproof}

\begin{lemma}[Dual balance and shared cache]
\label{lemma:dual_balance}
The iterates generated by Algorithm~\ref{alg:wq_gadmm} satisfy
\begin{equation}
\sum_{g=1}^G\lambda_g^k=0,
\qquad
w^k=\frac1G\sum_{g=1}^G w_g^k,
\qquad k\ge0.
\label{eq:nc_balance}
\end{equation}
After their respective refresh checks in iteration $k$, all groups use the same unquantized cache:
\[
z_g^k=\widehat z^k=w^{d^k},
\qquad
d_g^k=d^k,
\qquad
0\le k-d^k\le\tau_{\max}.
\]
Moreover,
\begin{equation}
\begin{aligned}
\widehat z^{k+1}&\in\{\widehat z^k,w^{k+1}\},\\
\|w^{k+1}-\widehat z^{k+1}\|^2
&\le\|w^{k+1}-\widehat z^k\|^2.
\end{aligned}
\label{eq:nc_cache_transition}
\end{equation}
\end{lemma}

\begin{IEEEproof}
The cloud and dual updates give
\[
\sum_g u_g^{k+1}
=\sum_g u_g^k+\sum_g w_g^{k+1}-Gw^{k+1}=0.
\]
Together with $u_g^0=0$ and $\lambda_g^k=\rho u_g^k$, this proves dual balance. Substitution into the cloud update proves the averaging identity, which also holds initially because $w_g^0=w^0$.

All groups start with the same cache and index. The common threshold and fixed cloud model within each window give
\[
(\widehat z^{k+1},d^{k+1})=
\begin{cases}
(\widehat z^k,d^k),&k+1-d^k\le\tau_{\max},\\
(w^{k+1},k+1),&k+1-d^k>\tau_{\max}.
\end{cases}
\]
Induction proves the shared-cache statement and its age bound. The squared distances in \eqref{eq:nc_cache_transition} are equal when the cache is retained; otherwise, the left-hand side is zero.
\end{IEEEproof}

The averaging identity and dual update also imply
\begin{equation}
\begin{aligned}
\sum_g r_g^k&=0,
&
\sum_g s_g^k&=Gv^k,\\
G\|v^k\|^2&\le\|s^k\|^2,
&
\lambda^{k+1}-\lambda^k&=\rho r^{k+1}.
\end{aligned}
\label{eq:nc_basic_identities}
\end{equation}

\begin{lemma}[Local update identity and error bound]
\label{lemma:perturbed_optimality}
Define the combined local and quantization error by
\begin{equation}
\begin{aligned}
\xi_g^{k+1}={}&e_g^{k+1}
+\nabla\phi_g(w_g^{k+1})
-\nabla\phi_g(\bar w_g^{k+1})\\
&+\rho\varepsilon_{d,g}^k
+(\rho+\eta)\varepsilon_{u,g}^{k+1},
\end{aligned}
\label{eq:nc_xi_def}
\end{equation}
and let $\xi^{k+1}=\operatorname{col}(\xi_1^{k+1},\ldots,\xi_G^{k+1})$. The local updates satisfy
\begin{equation}
\begin{aligned}
\nabla f(x^{k+1})+\lambda^{k+1}
&+\rho\mathbf1\otimes(w^{k+1}-\widehat z^k)\\
&+\eta s^{k+1}=\xi^{k+1}.
\end{aligned}
\label{eq:perturbed_optimality}
\end{equation}
Set
\begin{equation}
\begin{aligned}
A&=8\theta^2,\\
C&=4\max\{\rho^2,\;2\theta^2+L^2+(\rho+\eta)^2\}.
\end{aligned}
\label{eq:nc_error_constants}
\end{equation}
Under Assumptions~\ref{assu:nonconvex_smooth} and \ref{assu:bounded_quantization}, the local stopping criterion gives
\begin{equation}
\mathbb E\|\xi^{k+1}\|^2
\le AS_{k+1}+C\sigma_q^2.
\label{eq:nc_xi_bound}
\end{equation}
\end{lemma}

\begin{IEEEproof}
The local residual satisfies
\[
\begin{aligned}
e_g^{k+1}={}&\nabla\phi_g(\bar w_g^{k+1})\\
&+\rho(\bar w_g^{k+1}-\widehat z^k+u_g^k
-\varepsilon_{d,g}^k)\\
&+\eta(\bar w_g^{k+1}-w_g^k).
\end{aligned}
\]
Substituting $\bar w_g^{k+1}=w_g^{k+1}-\varepsilon_{u,g}^{k+1}$ and $w_g^{k+1}+u_g^k=u_g^{k+1}+w^{k+1}$ proves \eqref{eq:perturbed_optimality}.

The stopping criterion and smoothness give
\[
\|e_g^{k+1}\|^2
\le 2\theta_g^2\|s_g^{k+1}\|^2
+2\theta_g^2\|\varepsilon_{u,g}^{k+1}\|^2,
\]
and
\[
\|\nabla\phi_g(w_g^{k+1})-\nabla\phi_g(\bar w_g^{k+1})\|
\le L_g\|\varepsilon_{u,g}^{k+1}\|.
\]
Applying $\|\sum_{i=1}^4a_i\|^2\le4\sum_{i=1}^4\|a_i\|^2$ to \eqref{eq:nc_xi_def} yields
\[
\begin{aligned}
\|\xi_g^{k+1}\|^2\le{}&
8\theta_g^2\|s_g^{k+1}\|^2
+4\rho^2\|\varepsilon_{d,g}^k\|^2\\
&+4\bigl(2\theta_g^2+L_g^2+(\rho+\eta)^2\bigr)
\|\varepsilon_{u,g}^{k+1}\|^2.
\end{aligned}
\]
Summing over groups and taking expectations proves \eqref{eq:nc_xi_bound}.
\end{IEEEproof}

\begin{lemma}[Dual increment bound]
\label{lemma:nc_dual_increment}
For every $k\ge1$, the dual increment satisfies
\begin{equation}
\begin{aligned}
\|\lambda^{k+1}-\lambda^k\|^2
\le{}&4(L+\eta)^2\|s^{k+1}\|^2
+4\eta^2\|s^k\|^2\\
&+4\|\xi^{k+1}\|^2+4\|\xi^k\|^2.
\end{aligned}
\label{eq:nc_dual_bound}
\end{equation}
\end{lemma}

\begin{IEEEproof}
Let $P=(I_G-\mathbf1\mathbf1^\top/G)\otimes I_d$. Since $P\lambda^k=\lambda^k$ and $P(\mathbf1\otimes a)=0$, projecting \eqref{eq:perturbed_optimality} gives
\[
\lambda^{k+1}
=-P\nabla f(x^{k+1})-\eta Ps^{k+1}+P\xi^{k+1}.
\]
Subtracting the corresponding identity for $\lambda^k$ yields
\[
\begin{aligned}
\lambda^{k+1}-\lambda^k=-P\bigl[
&\nabla f(x^{k+1})-\nabla f(x^k)+\eta s^{k+1}\\
&-\eta s^k-\xi^{k+1}+\xi^k\bigr].
\end{aligned}
\]
The first three terms inside the brackets have combined norm at most $(L+\eta)\|s^{k+1}\|$. Using $\|Py\|\le\|y\|$ and the four-term squared-norm inequality proves \eqref{eq:nc_dual_bound}.
\end{IEEEproof}

Define the augmented Lagrangian
\begin{equation}
\begin{aligned}
\mathcal L_\rho(x,w,\lambda)
={}&f(x)+\langle\lambda,x-\mathbf1\otimes w\rangle+\frac{\rho}{2}\|x-\mathbf1\otimes w\|^2,
\end{aligned}
\label{eq:nc_lagrangian}
\end{equation}
and write $\mathcal L^k=\mathcal L_\rho(x^k,w^k,\lambda^k)$. To account for the cache difference and group model changes, define the potential function
\begin{equation}
\begin{aligned}
\Psi^k={}&\mathcal L^k
+\frac{\rho G}{2}\|w^k-\widehat z^k\|^2+\frac{4(\eta^2+A)}{\rho}\|s^k\|^2,
\qquad k\ge1.
\end{aligned}
\label{eq:Psi_def}
\end{equation}

\begin{lemma}[Potential decrease and lower bound]
\label{lemma:one_step_descent}
Suppose Assumptions~\ref{assu:nonconvex_smooth}-\ref{assu:finite_moments} hold. Define
\begin{equation}
\begin{aligned}
\kappa&=\frac{\rho}{4}+\eta-\frac L2
-\frac{4((L+\eta)^2+\eta^2)+9A}{\rho},\\
\nu&=\frac{9C}{\rho}.
\end{aligned}
\label{eq:nc_descent_constants}
\end{equation}
If
\begin{equation}
\rho>\frac{4L}{3},
\qquad
\kappa>0,
\label{eq:nc_parameter_condition}
\end{equation}
then, for all $k\ge1$,
\begin{equation}
\mathbb E[\Psi^{k+1}]
\le\mathbb E[\Psi^k]-\kappa S_{k+1}+\nu\sigma_q^2
\label{eq:main_descent}
\end{equation}
and
\begin{equation}
\mathbb E[\Psi^k]
\ge F_{\inf}-\frac{C}{\rho-L}\sigma_q^2.
\label{eq:nc_potential_lower}
\end{equation}
\end{lemma}

\begin{IEEEproof}
Smoothness gives the following bound on the change in the objective:
\[
f(x^{k+1})-f(x^k)
\le\langle\nabla f(x^{k+1}),s^{k+1}\rangle
+\frac L2\|s^{k+1}\|^2.
\]
Combining this inequality with the penalty difference and \eqref{eq:perturbed_optimality} gives
\begin{equation}
\begin{aligned}
&\mathcal L_\rho(x^{k+1},w^k,\lambda^k)
-\mathcal L_\rho(x^k,w^k,\lambda^k)\\
&\quad\le
-\frac{\rho+2\eta-L}{2}\|s^{k+1}\|^2
+\langle\xi^{k+1},s^{k+1}\rangle\\
&\qquad\quad
+\rho G\langle\widehat z^k-w^k,v^{k+1}\rangle,
\end{aligned}
\label{eq:nc_primal_descent}
\end{equation}
where the last term uses $\sum_gs_g^{k+1}=Gv^{k+1}$.

The cloud update minimizes a quadratic with Hessian $\rho G I_d$, decreasing the Lagrangian by $\rho G\|v^{k+1}\|^2/2$. The dual update increases it by $\|\lambda^{k+1}-\lambda^k\|^2/\rho$. Therefore,
\begin{equation}
\begin{aligned}
\mathcal L^{k+1}-\mathcal L^k\le{}&
-\frac{\rho+2\eta-L}{2}\|s^{k+1}\|^2
-\frac{\rho G}{2}\|v^{k+1}\|^2\\
&+\langle\xi^{k+1},s^{k+1}\rangle\\
&+\rho G\langle\widehat z^k-w^k,v^{k+1}\rangle\\
&+\frac1\rho\|\lambda^{k+1}-\lambda^k\|^2.
\end{aligned}
\label{eq:nc_lagrangian_difference}
\end{equation}
The cache term satisfies the identity
\begin{equation}
\begin{aligned}
2\langle\widehat z^k-w^k,v^{k+1}\rangle
={}&\|w^k-\widehat z^k\|^2
-\|w^{k+1}-\widehat z^k\|^2+\|v^{k+1}\|^2.
\end{aligned}
\label{eq:nc_cache_identity}
\end{equation}
Set $M^k=\mathcal L^k+\rho G\|w^k-\widehat z^k\|^2/2$. Substituting \eqref{eq:nc_cache_identity} into \eqref{eq:nc_lagrangian_difference} and applying \eqref{eq:nc_cache_transition} yields
\begin{equation}
\begin{aligned}
M^{k+1}-M^k\le{}&
-\frac{\rho+2\eta-L}{2}\|s^{k+1}\|^2
+\langle\xi^{k+1},s^{k+1}\rangle\\
&+\frac1\rho\|\lambda^{k+1}-\lambda^k\|^2.
\end{aligned}
\label{eq:nc_cache_potential_descent}
\end{equation}
Using $\langle\xi,s\rangle\le\rho\|s\|^2/4+\|\xi\|^2/\rho$ and \eqref{eq:nc_dual_bound}, we obtain
\begin{equation}
\begin{aligned}
M^{k+1}-M^k\le{}&
-\left(\frac{\rho}{4}+\eta-\frac L2
-\frac{4(L+\eta)^2}{\rho}\right)\|s^{k+1}\|^2\\
&+\frac{4\eta^2}{\rho}\|s^k\|^2
+\frac5\rho\|\xi^{k+1}\|^2+\frac4\rho\|\xi^k\|^2.
\end{aligned}
\label{eq:nc_pre_absorption}
\end{equation}
Taking expectations and applying \eqref{eq:nc_xi_bound} gives
\[
\begin{aligned}
\mathbb E[M^{k+1}]-\mathbb E[M^k]
\le{}&
-\left(\frac{\rho}{4}+\eta-\frac L2
-\frac{4(L+\eta)^2+5A}{\rho}\right)S_{k+1}\\
&+\frac{4(\eta^2+A)}{\rho}S_k
+\frac{9C}{\rho}\sigma_q^2.
\end{aligned}
\]
Adding $4(\eta^2+A)(S_{k+1}-S_k)/\rho$ proves \eqref{eq:main_descent}.

For the lower bound, smoothness implies
\[
F(w^k)\le
f(x^k)-\langle\nabla f(x^k),r^k\rangle
+\frac L2\|r^k\|^2.
\]
Hence
\[
\mathcal L^k\ge
F(w^k)+\langle\nabla f(x^k)+\lambda^k,r^k\rangle
+\frac{\rho-L}{2}\|r^k\|^2.
\]
Equation~\eqref{eq:perturbed_optimality} at the preceding update gives
\[
\nabla f(x^k)+\lambda^k
=\xi^k-\eta s^k-\rho\mathbf1\otimes(w^k-\widehat z^{k-1}).
\]
The cache term is orthogonal to $r^k$ because $\sum_g r_g^k=0$. Completing the square therefore yields
\[
\begin{aligned}
\mathcal L^k
&\ge F_{\inf}
-\frac{\|\xi^k-\eta s^k\|^2}{2(\rho-L)}\\
&\ge F_{\inf}
-\frac{\|\xi^k\|^2+\eta^2\|s^k\|^2}{\rho-L}.
\end{aligned}
\]
Using \eqref{eq:nc_xi_bound} and the nonnegative cache term in $\Psi^k$ gives
\[
\begin{aligned}
\mathbb E[\Psi^k]\ge{}&
F_{\inf}-\frac{C}{\rho-L}\sigma_q^2\\
&+(A+\eta^2)
\left(\frac4\rho-\frac1{\rho-L}\right)S_k.
\end{aligned}
\]
The last coefficient is nonnegative when $\rho>4L/3$, proving \eqref{eq:nc_potential_lower}.
\end{IEEEproof}

The parameter condition also implies $\rho+\eta>L$, ensuring finite local updates by Lemma~\ref{lemma:local_termination}. One sufficient choice is $\rho\ge32(L+\eta+\theta)$, which gives $\kappa\ge5\rho/32+\eta>0$.

\subsection{Average KKT Residual Bound}

Define the squared KKT residual as
\begin{equation}
\mathcal R^k
=\|\nabla f(x^k)+\lambda^k\|^2+\|r^k\|^2,
\qquad k\ge1.
\label{eq:nc_kkt_residual}
\end{equation}
Its two terms measure stationarity and consensus errors. The remaining KKT condition, $\sum_g\lambda_g^k=0$, holds exactly by Lemma~\ref{lemma:dual_balance}.

\begin{theorem}[Average KKT residual bound]
\label{thm:bounded_quant_convergence}
Suppose Assumptions~\ref{assu:nonconvex_smooth}-\ref{assu:finite_moments} hold and the parameters satisfy \eqref{eq:nc_parameter_condition}. Set $H=\tau_{\max}+1$ and
\begin{equation}
\begin{aligned}
B&=\mathbb E[\Psi^1]-F_{\inf}
+\frac{C}{\rho-L}\sigma_q^2,\\
T&=3(A+\eta^2)
+\frac{4((L+\eta)^2+\eta^2+2A)}{\rho^2}
+3\rho^2H^2,\\
\widehat C_0&=\mathbb E[\mathcal R^1]
+T\left(S_1+\frac B\kappa\right),\\
C_q&=\frac{T\nu}{\kappa}
+C\left(3+\frac8{\rho^2}\right).
\end{aligned}
\label{eq:nc_rate_constants}
\end{equation}
These constants are finite and independent of $K$, with $B\ge0$. For every integer $K\ge1$,
\begin{equation}
\frac1K\sum_{k=1}^K\mathbb E[\mathcal R^k]
\le\frac{\widehat C_0}{K}+C_q\sigma_q^2.
\label{eq:rate_bound}
\end{equation}
Consequently,
\begin{equation}
\limsup_{K\to\infty}
\frac1K\sum_{k=1}^K\mathbb E[\mathcal R^k]
\le C_q\sigma_q^2.
\label{eq:limsup_bound}
\end{equation}
\end{theorem}

\begin{IEEEproof}
For any integer $n\ge1$, summing \eqref{eq:main_descent} and applying \eqref{eq:nc_potential_lower} gives
\begin{equation}
\sum_{k=1}^n S_{k+1}
\le\frac B\kappa+\frac{\nu n}{\kappa}\sigma_q^2.
\label{eq:nc_increment_sum}
\end{equation}
The potential lower bound at $k=1$ also implies $B\ge0$. Finite second moments and smoothness ensure that the constants in \eqref{eq:nc_rate_constants} are finite.

Equations~\eqref{eq:nc_basic_identities}, \eqref{eq:nc_dual_bound}, and \eqref{eq:nc_xi_bound} imply, for $k\ge1$,
\begin{equation}
\begin{aligned}
\mathbb E\|r^{k+1}\|^2\le{}&
\frac{4((L+\eta)^2+A)}{\rho^2}S_{k+1}\\
&+\frac{4(\eta^2+A)}{\rho^2}S_k
+\frac{8C}{\rho^2}\sigma_q^2.
\end{aligned}
\label{eq:nc_consensus_bound}
\end{equation}
Similarly, the local update identity gives
\begin{equation}
\begin{aligned}
\mathbb E\|\nabla f(x^{k+1})+\lambda^{k+1}\|^2
\le{}&3(A+\eta^2)S_{k+1}+3C\sigma_q^2\\
&+3\rho^2G\mathbb E\|w^{k+1}-\widehat z^k\|^2.
\end{aligned}
\label{eq:nc_stationarity_bound}
\end{equation}

Since $\widehat z^k=w^{d^k}$ and $k-d^k\le\tau_{\max}$, the difference $w^{k+1}-\widehat z^k$ contains at most $H$ consecutive increments. Cauchy-Schwarz and $G\|v^j\|^2\le\|s^j\|^2$ give
\begin{equation}
G\|w^{k+1}-\widehat z^k\|^2
\le H
\sum_{j=\max\{1,k-\tau_{\max}+1\}}^{k+1}
\|s^j\|^2.
\label{eq:nc_cache_residual_bound}
\end{equation}
Each increment appears in at most $H$ such sums. Therefore,
\begin{equation}
\sum_{k=1}^n
G\mathbb E\|w^{k+1}-\widehat z^k\|^2
\le H^2
\left(S_1+\sum_{k=1}^nS_{k+1}\right).
\label{eq:nc_cache_sum}
\end{equation}
Using $\sum_{k=1}^nS_k\le S_1+\sum_{k=1}^nS_{k+1}$ and combining the preceding bounds yields
\begin{equation}
\begin{aligned}
\sum_{k=1}^n\mathbb E[\mathcal R^{k+1}]
\le{}&
T\left(S_1+\sum_{k=1}^nS_{k+1}\right)\\
&+nC\left(3+\frac8{\rho^2}\right)\sigma_q^2\\
\le{}&
T\left(S_1+\frac B\kappa\right)+nC_q\sigma_q^2.
\end{aligned}
\label{eq:nc_residual_sum}
\end{equation}

For $K\ge2$, set $n=K-1$ and add $\mathbb E[\mathcal R^1]$ to obtain
\[
\sum_{k=1}^K\mathbb E[\mathcal R^k]
\le\widehat C_0+(K-1)C_q\sigma_q^2.
\]
Dividing by $K$ proves \eqref{eq:rate_bound}. The case $K=1$ follows from $\widehat C_0\ge\mathbb E[\mathcal R^1]$, and taking the limit superior proves \eqref{eq:limsup_bound}.
\end{IEEEproof}

For fixed algorithm parameters, \eqref{eq:limsup_bound} bounds the limiting average squared residual in proportion to the quantization error bound. In the absence of quantization error, $\sigma_q^2=0$, equation~\eqref{eq:rate_bound} implies
\[
\sum_{k=1}^{\infty}\mathbb E[\mathcal R^k]
\le\widehat C_0<\infty,
\qquad
\mathbb E[\mathcal R^k]\to0.
\]
With persistent quantization errors, the result bounds the average residual but does not guarantee convergence of the full model sequence or global optimality.

\section{Experiments}
\label{sec:experiments}

In this section, we evaluate WQ-GADMM on a smooth nonconvex test problem and on the MNIST and CIFAR-10 datasets. We first describe the experimental settings, then examine residual behaviour and quantization precision. Finally, we compare learning performance, group participation, and resource efficiency with the baseline methods.

\subsection{Experimental Setup}
\label{subsec:setup}

The following settings apply to all experiments on MNIST and CIFAR-10. The smooth nonconvex experiment is described separately in Section~\ref{subsec:residual_experiment}.

\textbf{(a) Platform and datasets:}
Experiments were conducted on a workstation equipped with an NVIDIA GeForce RTX 4090 GPU, using Python 3.10 and PyTorch 2.10.0 with CUDA 13.0. We simulated a network with 50 clients divided into five fixed edge groups of ten clients each. Convolutional neural networks were trained on MNIST and CIFAR-10. Client data were partitioned using a Dirichlet distribution with parameter $\alpha=0.1$. For each random seed, all methods shared the same model initialization, data partition, client grouping, and computation profile. Each setting was evaluated over three random seeds; results with error terms are reported as mean $\pm$ standard deviation.

\textbf{(b) Training settings:}
The training workload was fixed at 25,000 cumulative client-gradient evaluations for MNIST and 1,500,000 for CIFAR-10. Local updates used a batch size of 64 and at most two stochastic gradient steps. At most $M_a=2$ groups were active concurrently, and model staleness did not exceed $\tau_{\max}=1$ for any method. P-GADMM and P-GADMM-Sync both used a learning rate of 0.01 and an ADMM penalty parameter of 0.001 on both datasets. For WQ-GADMM, the learning rate, penalty parameter $\rho$, and proximal parameter $\eta$ were 0.20, 0.003, and 0.005 on MNIST, and 0.03, 0.02, and 0.01 on CIFAR-10, respectively. For both datasets, WQ-GADMM used a local stopping-check parameter of 0.05. Its scheduling parameters were $T_{\mathrm{act}}=3$, $\omega_1=1$, $\omega_2=0.2$, and $\epsilon_s=10^{-12}$. Ties were resolved in ascending order of group index.

\textbf{(c) Compared methods:}
We compared P-GADMM \cite{Wei2026PGADMM}, P-GADMM-Sync, WQ-GADMM-FP32, and WQ-GADMM-Q12.

P-GADMM executed at most two group tasks concurrently. The cloud performed a global update upon receiving two completed group results, retaining the previous models of the remaining groups. Clients performed local stochastic gradient descent (SGD), followed by model averaging within each group. P-GADMM-Sync retained the local updates, optimization parameters, and 32-bit communication of P-GADMM. It used a round-robin schedule with a $2+2+1$ activation pattern, with each group updating once before the global and dual updates at the end of the window. WQ-GADMM-FP32 and WQ-GADMM-Q12 used the proposed algorithm with identical optimization parameters and group update sequences. FP32 used 32-bit communication, whereas Q12 applied 12-bit tensor-wise symmetric stochastic quantization to both downlink and uplink exchanges. Both configurations completed one update per group within each $2+2+1$ window.

\textbf{(d) Metrics and simulation settings:}
We measured test accuracy, communication volume, simulated wall-clock time, the Jain fairness index, full-group coverage, and mean inter-completion gaps. Communication volume included cloud-edge model exchanges and excluded communication between clients and edge servers. Simulated wall-clock time accounted for local computation, model transmission, and waiting during concurrent group execution. Client computation heterogeneity followed a Pareto distribution with shape parameter 1.1. Downlink and uplink rates were fixed at 5~Mbit/s and 2~Mbit/s, respectively. The same computation and communication timing rules were applied to all methods, with waiting times determined by their respective update schedules.

\subsection{Residual Behavior on a Smooth Nonconvex Problem}
\label{subsec:residual_experiment}

We examined the effect of communication precision on a smooth nonconvex problem by measuring the per-iteration squared KKT residual and its cumulative average. The problem consisted of five groups and a 12-dimensional model, with group objectives
\begin{equation}
\phi_g(w)=
\sum_{j=1}^{12}
\left[
\frac{q_{gj}}{2}
\left(w_j-\frac{b_{gj}}{q_{gj}}\right)^2
+a(1-\cos w_j)
\right],
\label{eq:synthetic_objective}
\end{equation}
where $q_{gj}\sim\mathcal U(0.15,0.35)$, $b_{gj}\sim\mathcal U(-0.20,0.20)$, and $a=0.8$. The coefficients were fixed across all runs. Each objective is bounded below and has a globally Lipschitz continuous gradient with $L_g=\max_j q_{gj}+a$. Its Hessian is negative definite at $w=\pi\mathbf1$, confirming nonconvexity.

We compared FP32, Q12, and Q8 with identical objectives and algorithm parameters. We set $\rho=12$, $\eta=0.2$, and $\theta_g=0.1$, satisfying \eqref{eq:nc_parameter_condition}. The scheduling and staleness settings were $M_a=2$, $T_{\mathrm{act}}=2$, $\tau_{\max}=2$, and $\omega_1=\omega_2=1$. Local updates used full gradients, stepsize $\zeta_g=(L_g+\rho+\eta)^{-1}$, and the stopping criterion in \eqref{eq:local_residual_bound}.

Each configuration ran for 1,000 logical iterations over three random seeds. For each seed, all groups and configurations shared an initial model with coordinates sampled from $\mathcal U(0.6,1.0)$, and all dual variables were initialized to zero. Q12 and Q8 used stochastic rounding on $2^{12}$ and $2^8$ uniformly spaced levels over $[-2,2]$, respectively, for both communication directions. No quantizer input exceeded this range during the runs. After each global and dual update, we computed the squared KKT residual $\mathcal R^k$ defined in \eqref{eq:nc_kkt_residual} and its cumulative average, $\overline{\mathcal R}_K=K^{-1}\sum_{k=1}^{K}\mathcal R^k$.

\begin{figure}[t]
\centering
\subfloat[Per-iteration KKT residual]{%
\includegraphics[width=0.50\linewidth]{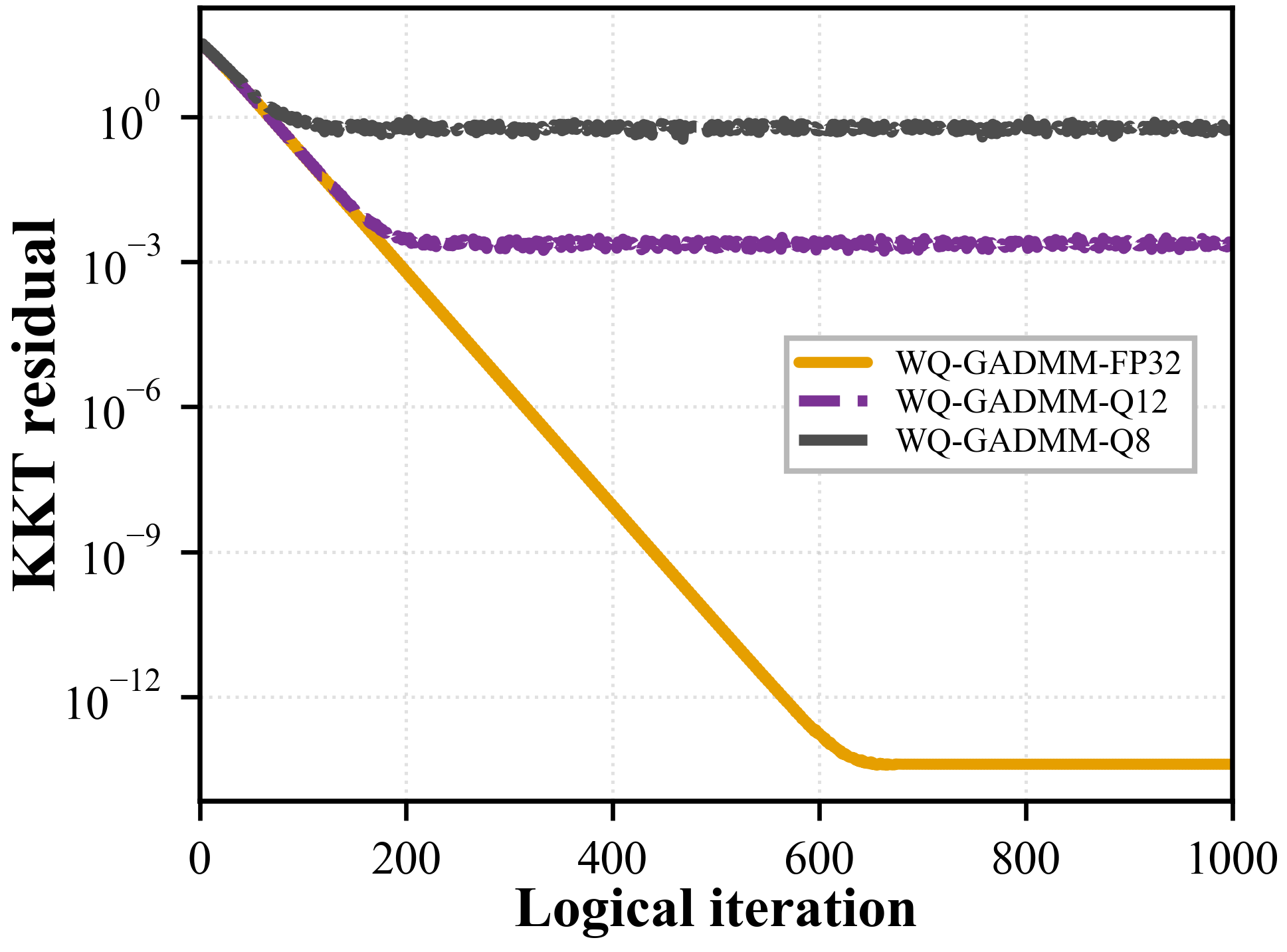}}
\hfill
\subfloat[Average KKT residual]{%
\includegraphics[width=0.50\linewidth]{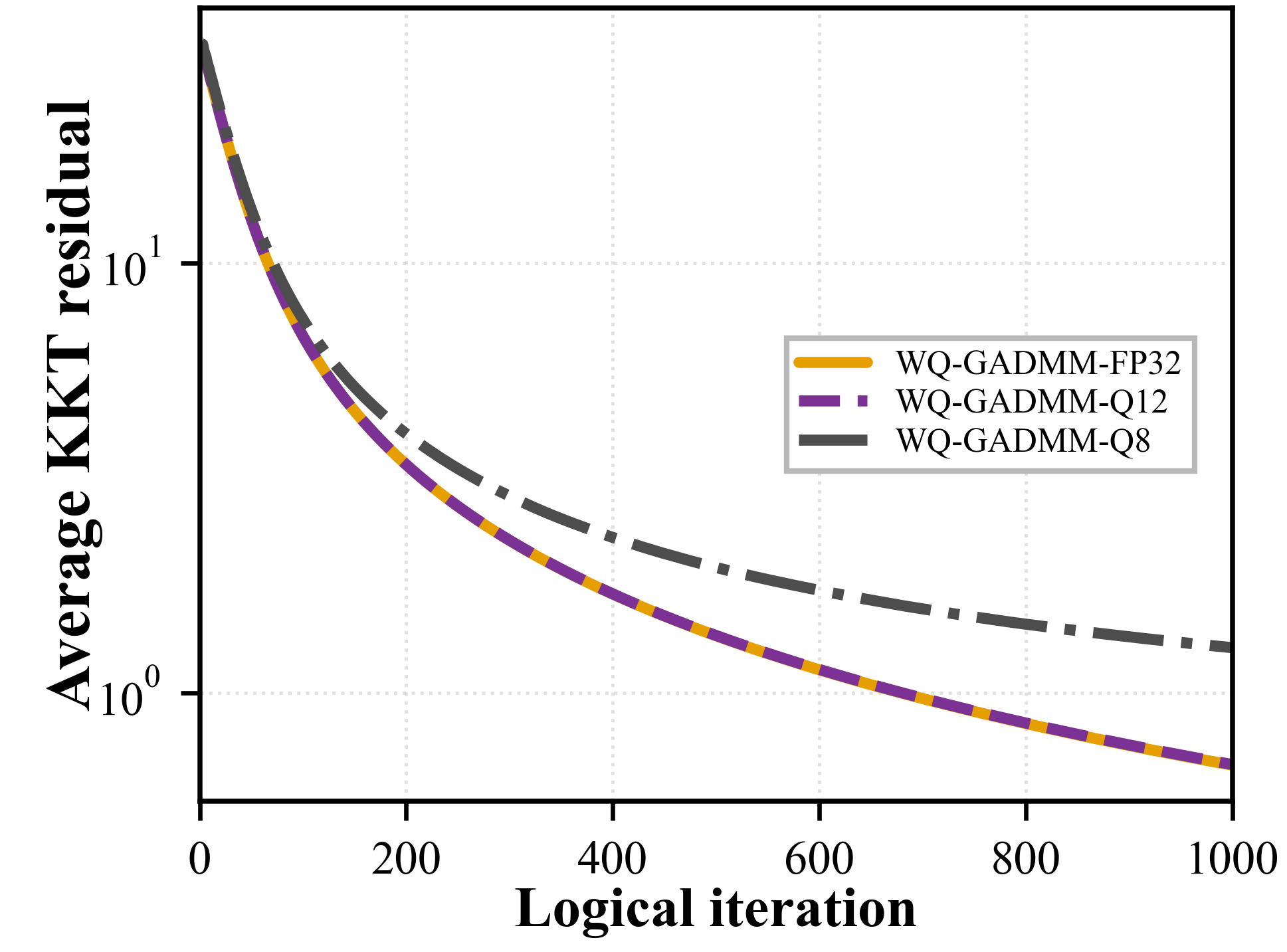}}
\caption{Per-iteration and cumulative average squared KKT residuals on the smooth nonconvex test problem.}
\label{fig:nonconvex_residual}
\end{figure}

Fig.~\ref{fig:nonconvex_residual}(a) shows that the FP32 residual decreased to a low numerical level, whereas Q12 and Q8 fluctuated around higher levels. Over the final 200 iterations and three seeds, the mean Q12 residual was approximately $1/247$ of that for Q8, indicating that finer quantization reduced the residual level in this experiment. In Fig.~\ref{fig:nonconvex_residual}(b), the cumulative average residual decreased for all configurations. FP32 and Q12 nearly overlapped because early iterations dominated their cumulative averages. These results are consistent with the average residual bound and its quantization-dependent error term.

\begin{table*}[t]
\centering
\caption{Effect of quantization precision on WQ-GADMM.}
\label{tab:quantization}
{\footnotesize
\setlength{\tabcolsep}{8pt}
\renewcommand{\arraystretch}{1.12}
\begin{tabular}{cccccc}
\toprule
\makecell{Quantization\\precision} &
\makecell{MNIST\\accuracy (\%)} &
\makecell{CIFAR-10\\accuracy (\%)} &
\makecell{Communication\\reduction (\%)} &
\makecell{MNIST\\consensus residual} &
\makecell{CIFAR-10\\consensus residual} \\
\midrule
FP32 & $95.86 \pm 0.40$ & $67.47 \pm 0.67$ & 0.00 & 0.0903 & \textbf{0.0182} \\
Q16 & $95.85 \pm 0.58$ & $67.48 \pm 0.65$ & 50.00 & 0.0919 & 0.0183 \\
Q12 & $\mathbf{96.05 \pm 0.27}$ & $\mathbf{67.58 \pm 0.63}$ & 62.50 & \textbf{0.0836} & 0.0255 \\
Q8 & $96.01 \pm 0.27$ & $67.57 \pm 0.91$ & \textbf{75.00} & 0.1436 & 0.2509 \\
Q2 & Diverged & Diverged & $93.75^{\dagger}$ & - & - \\
\bottomrule
\end{tabular}

\vspace{1mm}
\parbox{0.93\textwidth}{\footnotesize
$^{\dagger}$Q2 provides a theoretical reduction of 93.75\%, but all runs diverged before completing training.}
}
\end{table*}

\subsection{Effect of Quantization Precision}

We next compared FP32, Q16, Q12, Q8, and Q2 on MNIST and CIFAR-10 to select the communication precision for subsequent experiments. All configurations shared the same data partition, initialization, client grouping, prescribed workload, and group update sequence. Only the precision of the transmitted models varied.

Table~\ref{tab:quantization} shows that Q16, Q12, and Q8 retained final accuracies close to FP32. However, Q8 produced larger consensus residuals than Q12, particularly on CIFAR-10. Q2 diverged on both datasets.We therefore selected Q12 for the remaining comparisons.

\subsection{Overall Performance}

We compared the four methods under the same training workload to evaluate learning performance, communication cost, and simulated completion time.

\begin{figure}[t]
\centering
\subfloat[MNIST]{%
\includegraphics[width=0.50\linewidth]{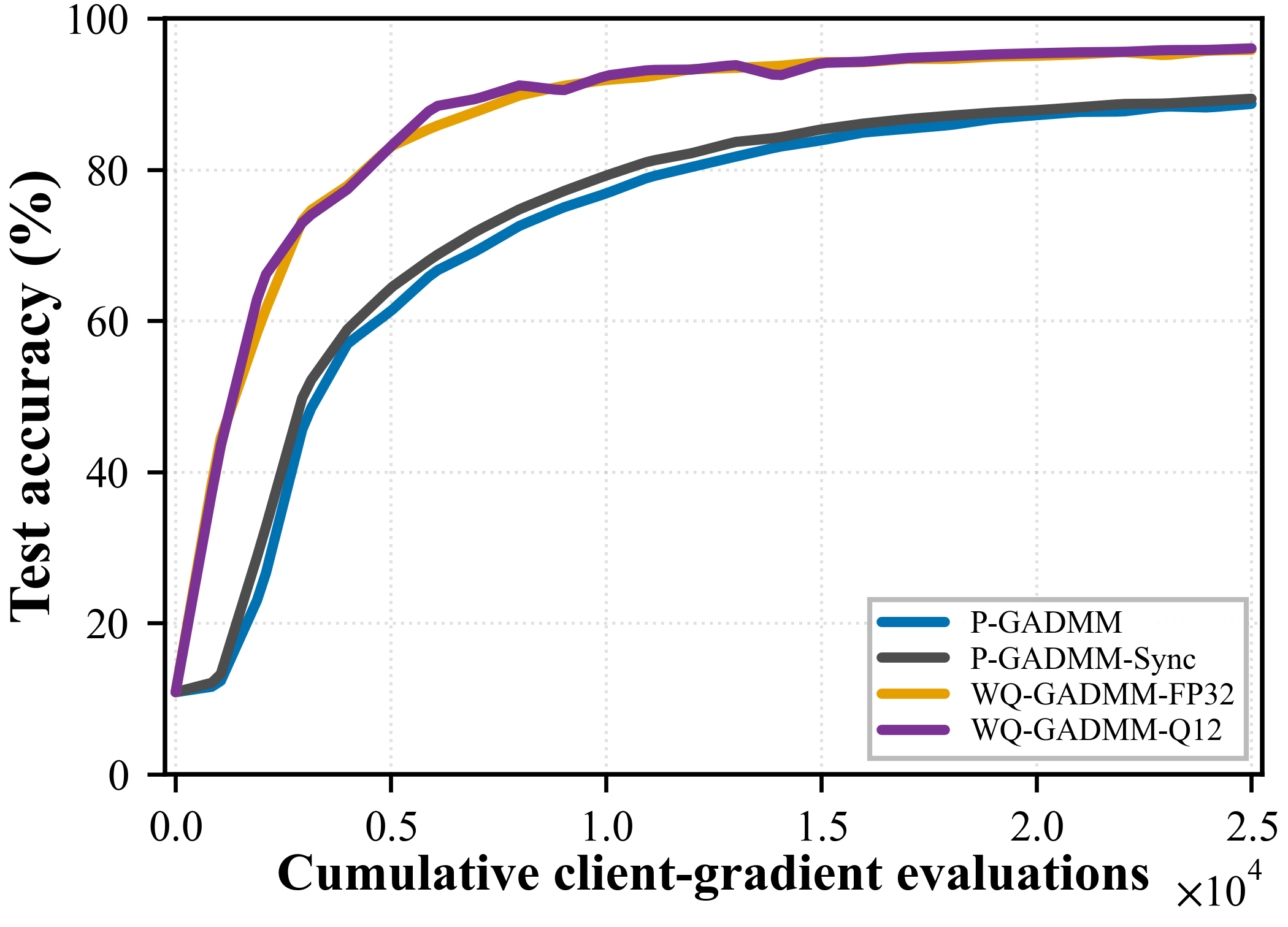}}
\hfill
\subfloat[CIFAR-10]{%
\includegraphics[width=0.50\linewidth]{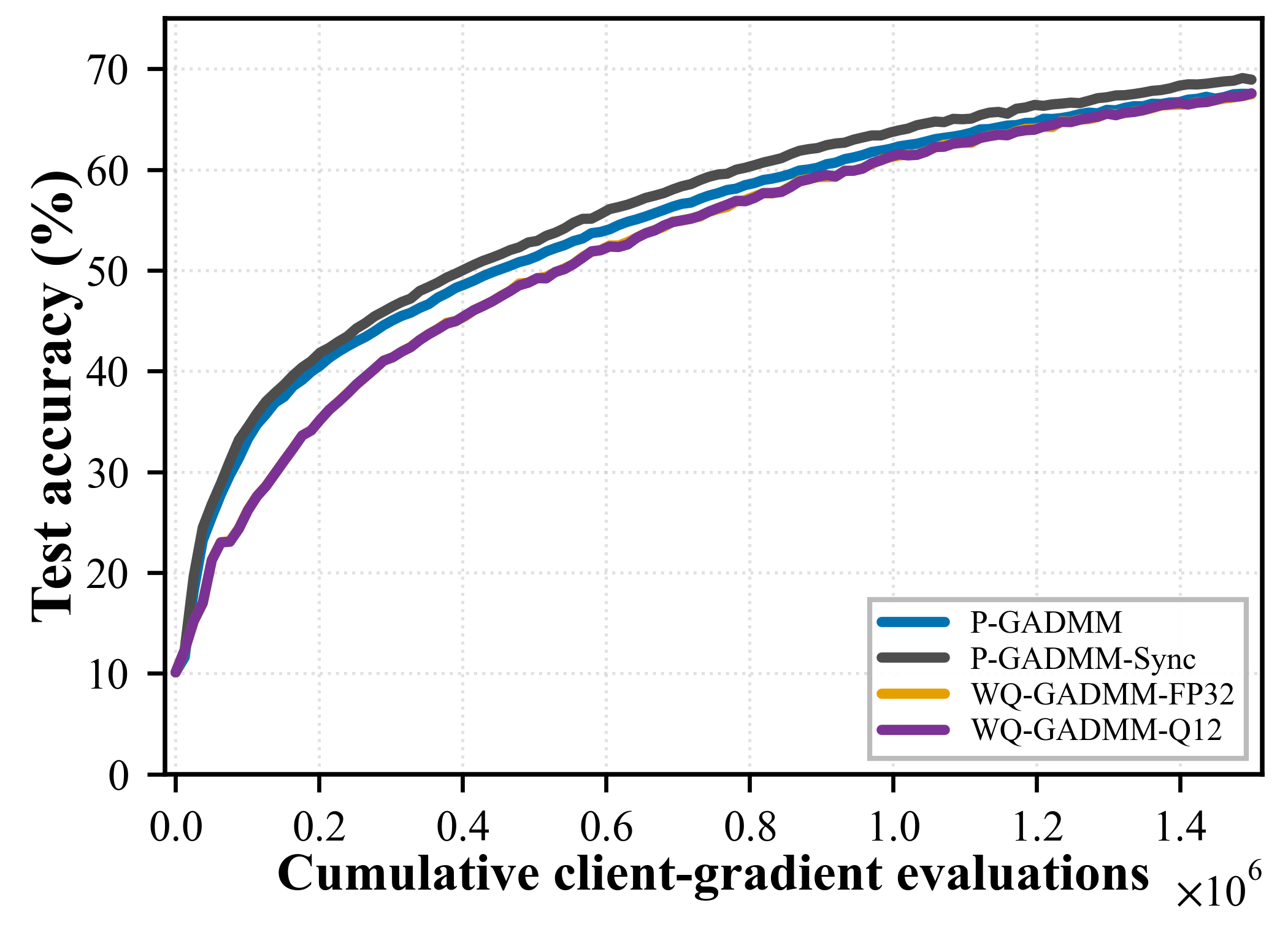}}
\caption{Test accuracy versus cumulative client-gradient evaluations.}
\label{fig:gradient}
\end{figure}

\begin{figure}[t]
\centering
\subfloat[MNIST]{%
\includegraphics[width=0.50\linewidth]{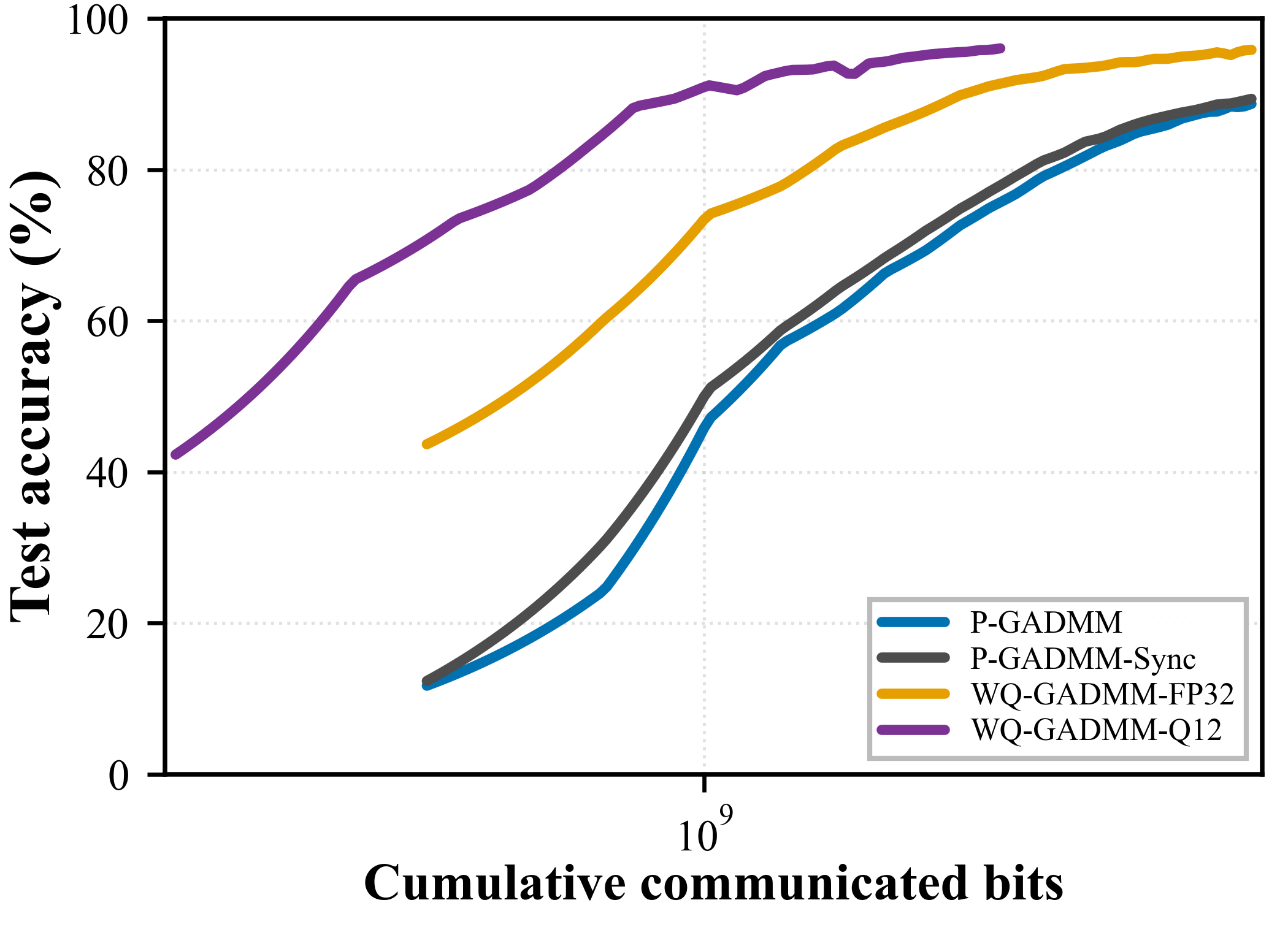}}
\hfill
\subfloat[CIFAR-10]{%
\includegraphics[width=0.50\linewidth]{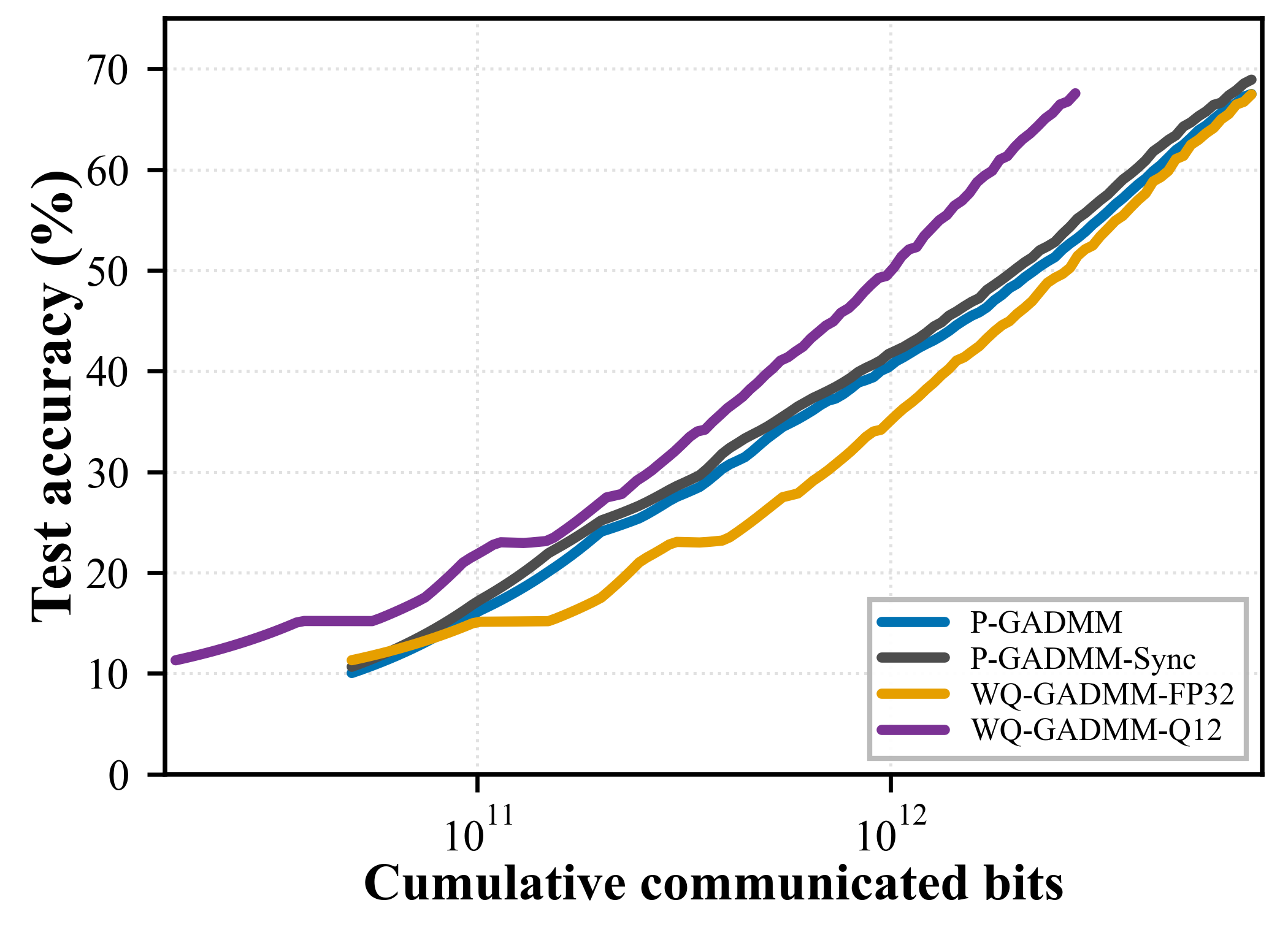}}
\caption{Test accuracy versus cumulative communicated bits.}
\label{fig:communication}
\end{figure}

\begin{figure}[t]
\centering
\subfloat[MNIST]{%
\includegraphics[width=0.50\linewidth]{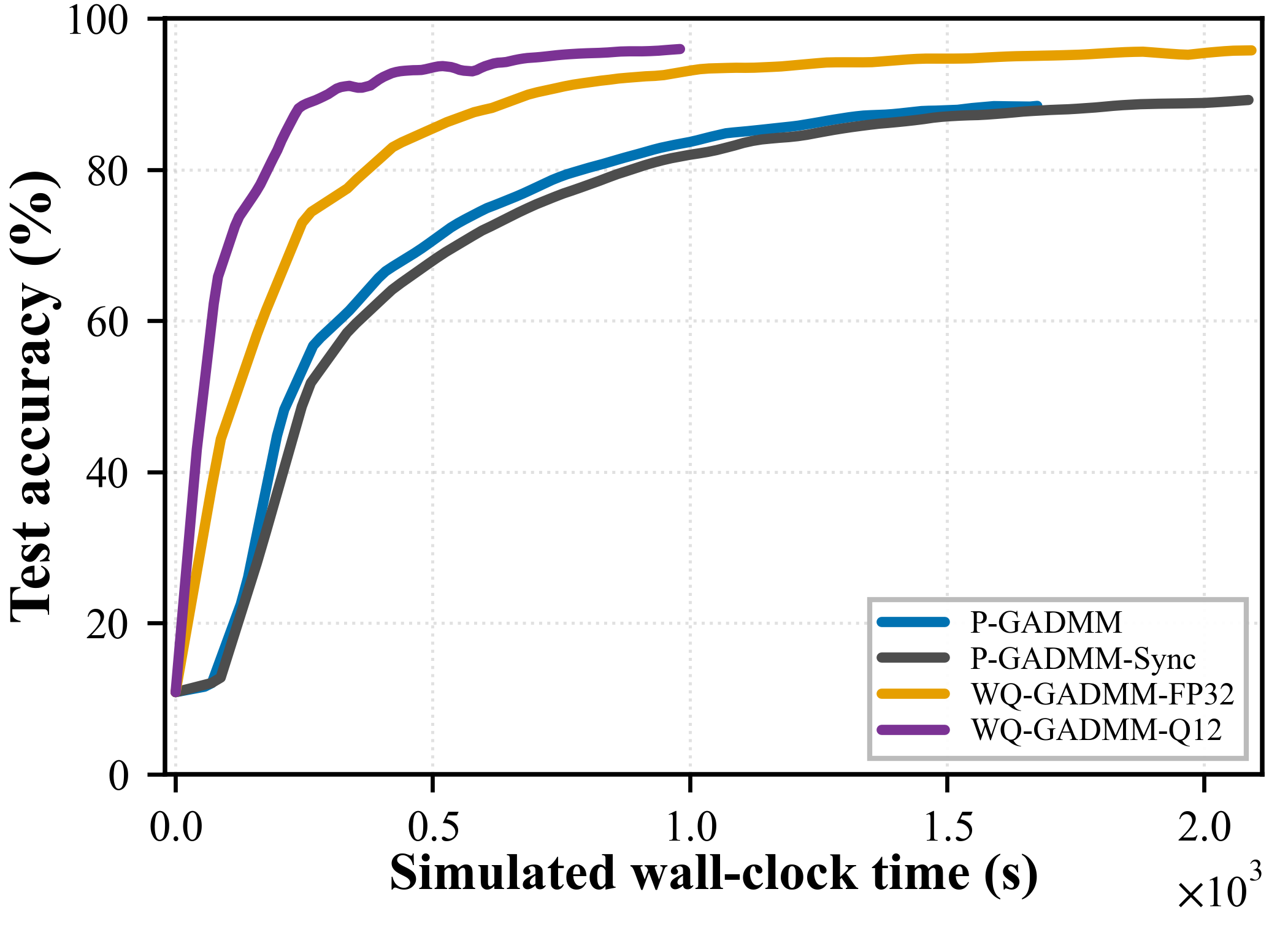}}
\hfill
\subfloat[CIFAR-10]{%
\includegraphics[width=0.50\linewidth]{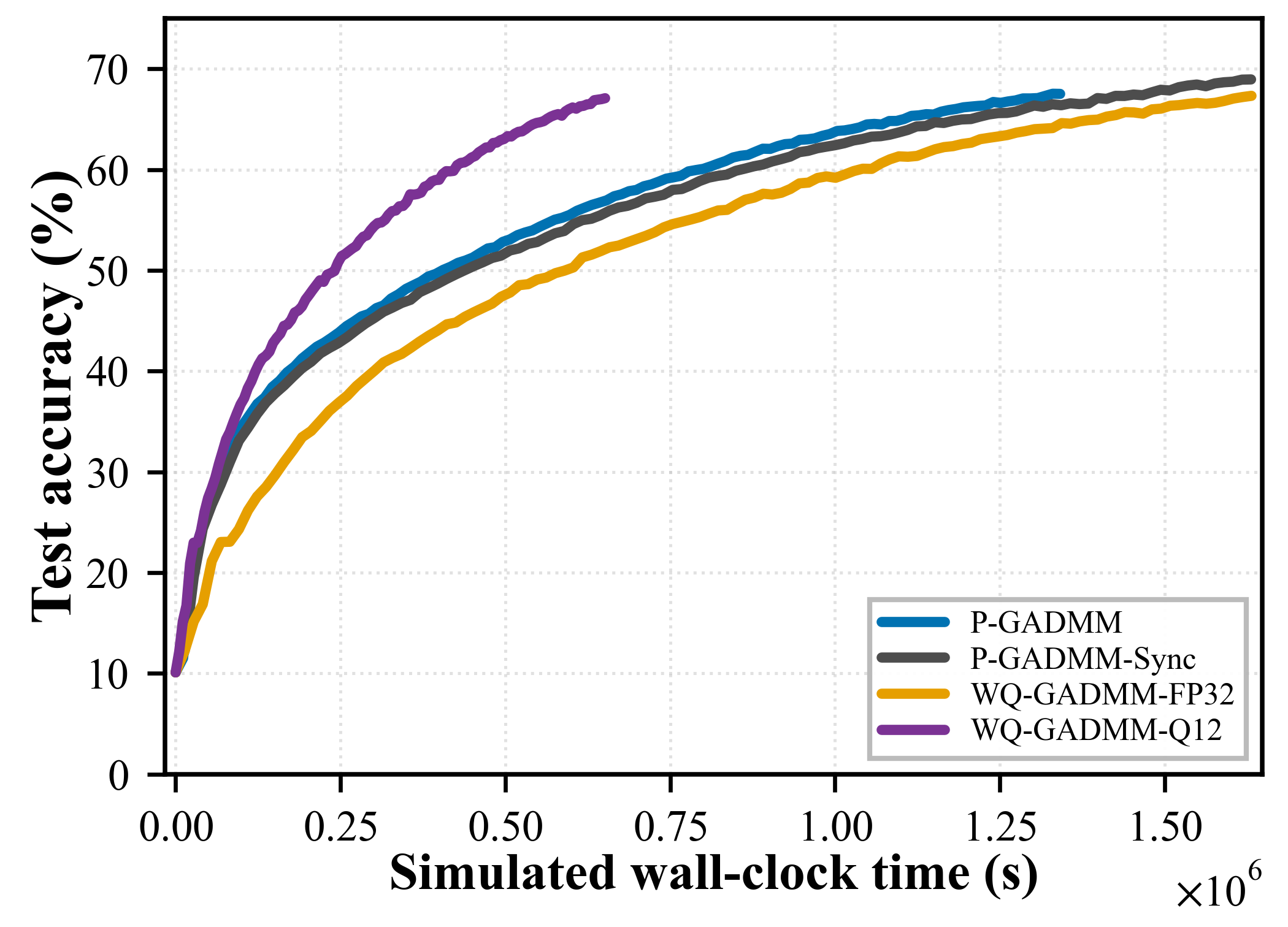}}
\caption{Test accuracy versus simulated wall-clock time.}
\label{fig:wallclock}
\end{figure}

\begin{table*}[t]
\centering
\caption{Final accuracy and resource consumption under a fixed training workload.}
\label{tab:overall}
{\footnotesize
\setlength{\tabcolsep}{9pt}
\renewcommand{\arraystretch}{1.12}
\begin{tabular}{clccc}
\toprule
Dataset & Method &
\makecell{Test accuracy\\(\%)} &
\makecell{Cumulative communication\\(Gbit)} &
\makecell{Simulated wall-clock\\time (s)} \\
\midrule
\multirow{4}{*}{MNIST}
& P-GADMM & $88.70 \pm 0.39$ & 8.469 & 1694.7 \\
& P-GADMM-Sync & $89.38 \pm 0.35$ & 8.469 & 2131.0 \\
& WQ-GADMM-FP32 & $95.86 \pm 0.40$ & 8.469 & 2113.8 \\
& WQ-GADMM-Q12 & $\mathbf{96.05 \pm 0.27}$ & \textbf{3.177} & \textbf{1002.3} \\
\midrule
\multirow{4}{*}{CIFAR-10}
& P-GADMM & $67.51 \pm 0.85$ & $7.469\times10^3$ & $1.352\times10^6$ \\
& P-GADMM-Sync & $\mathbf{68.93 \pm 0.57}$ & $7.469\times10^3$ & $1.645\times10^6$ \\
& WQ-GADMM-FP32 & $67.47 \pm 0.67$ & $7.469\times10^3$ & $1.647\times10^6$ \\
& WQ-GADMM-Q12 & $67.58 \pm 0.63$ & $\mathbf{2.801\times10^3}$ & $\mathbf{0.667\times10^6}$ \\
\bottomrule
\end{tabular}
}
\end{table*}

Table~\ref{tab:overall} compares final accuracy and resource consumption under the same training workload. Both WQ-GADMM configurations achieved higher accuracy than the baselines on MNIST, while all four methods achieved similar final accuracy on CIFAR-10. WQ-GADMM-Q12 required the lowest communication volume and shortest simulated wall-clock time on both datasets. Compared with WQ-GADMM-FP32, it reduced communication volume by 62.5\% while maintaining comparable final accuracy. Fig.~\ref{fig:gradient} shows a similar pattern during training: WQ-GADMM achieved higher accuracy at comparable numbers of client-gradient evaluations on MNIST, whereas the accuracy curves were close on CIFAR-10. Figs.~\ref{fig:communication} and \ref{fig:wallclock} show that Q12 also reached accuracy comparable to FP32 with fewer transmitted bits and less simulated time.

\subsection{Group Participation and Inter-Completion Gaps}

We evaluated whether groups participated regularly and how long they waited between completed updates. The comparison included P-GADMM, P-GADMM-Sync, and WQ-GADMM-Q12. Each observation interval contained five completed group updates, with 30 intervals for MNIST and 20 for CIFAR-10. The full-group coverage ratio was the fraction of intervals in which all five groups completed an update.

\begin{table*}[t]
\centering
\caption{Group participation and inter-completion gaps.}
\label{tab:fairness}
{\footnotesize
\setlength{\tabcolsep}{11pt}
\renewcommand{\arraystretch}{1.12}
\begin{tabular}{clccc}
\toprule
Dataset & Method &
\makecell{Jain fairness\\index} &
\makecell{Full-group\\coverage ratio} &
\makecell{Mean inter-completion\\gap (s)} \\
\midrule
\multirow{3}{*}{MNIST}
& P-GADMM & 0.963 & 0.522 & 7.03 \\
& P-GADMM-Sync & 1.000 & 1.000 & 8.53 \\
& WQ-GADMM-Q12 & 1.000 & 1.000 & \textbf{4.02} \\
\midrule
\multirow{3}{*}{CIFAR-10}
& P-GADMM & 0.965 & 0.533 & 93.28 \\
& P-GADMM-Sync & 1.000 & 1.000 & 109.70 \\
& WQ-GADMM-Q12 & 1.000 & 1.000 & \textbf{44.38} \\
\bottomrule
\end{tabular}
}
\end{table*}

\begin{figure}[t]
\centering
\subfloat[MNIST]{%
\includegraphics[width=0.49\linewidth]{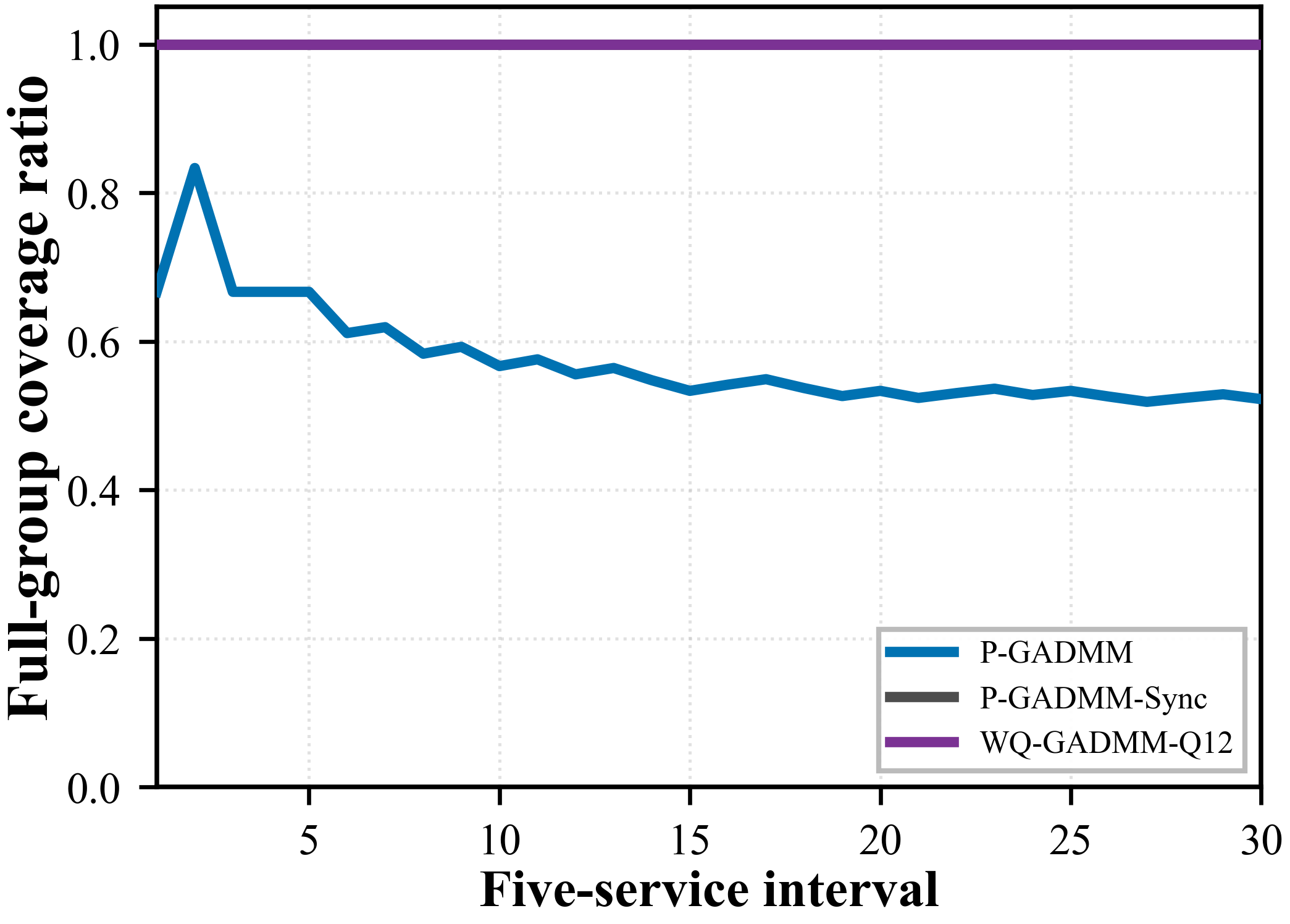}}
\hfill
\subfloat[CIFAR-10]{%
\includegraphics[width=0.49\linewidth]{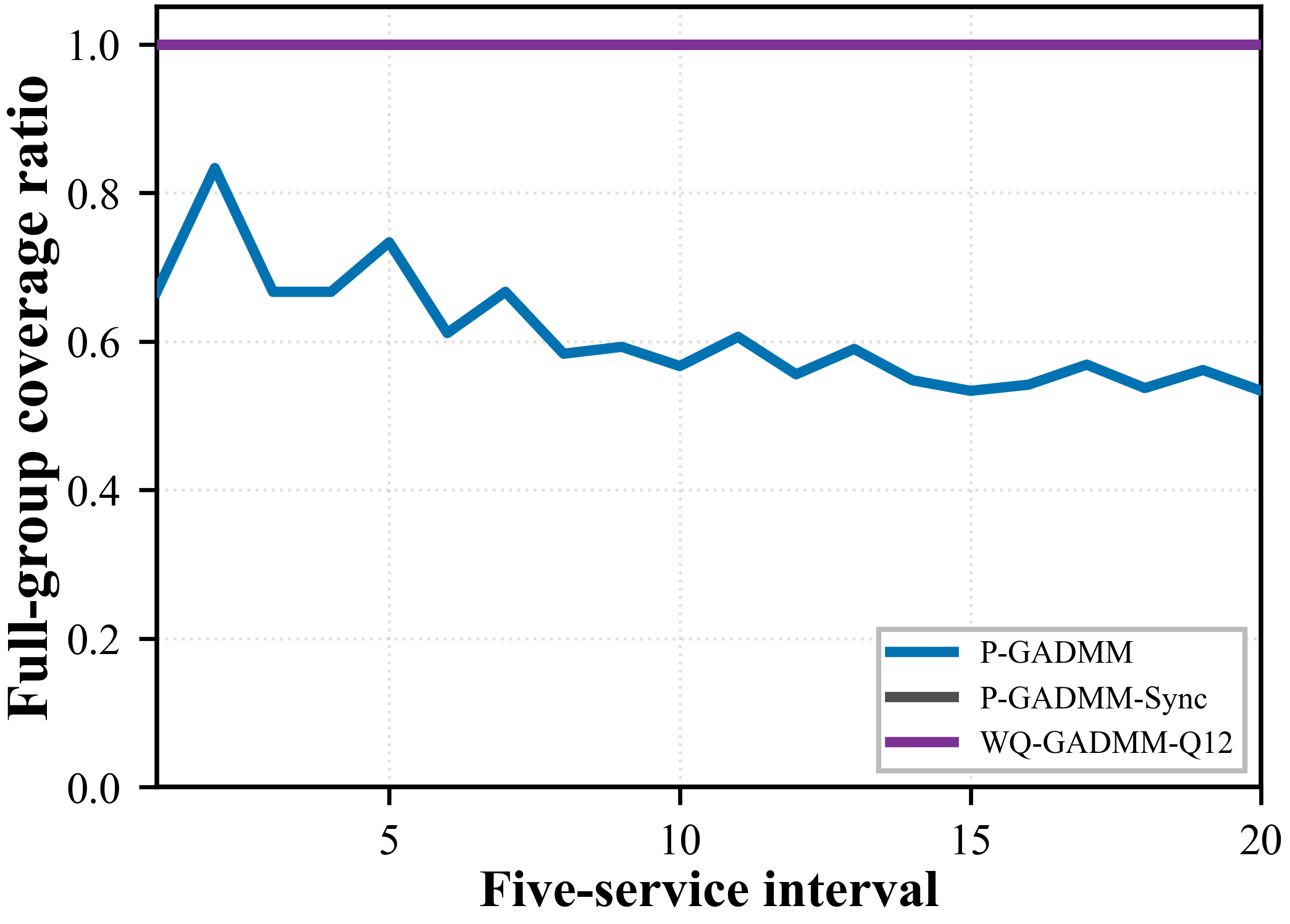}}
\caption{Cumulative full-group coverage ratio.}
\label{fig:coverage}
\end{figure}

\begin{figure}[t]
\centering
\subfloat[MNIST]{%
\includegraphics[width=0.49\linewidth]{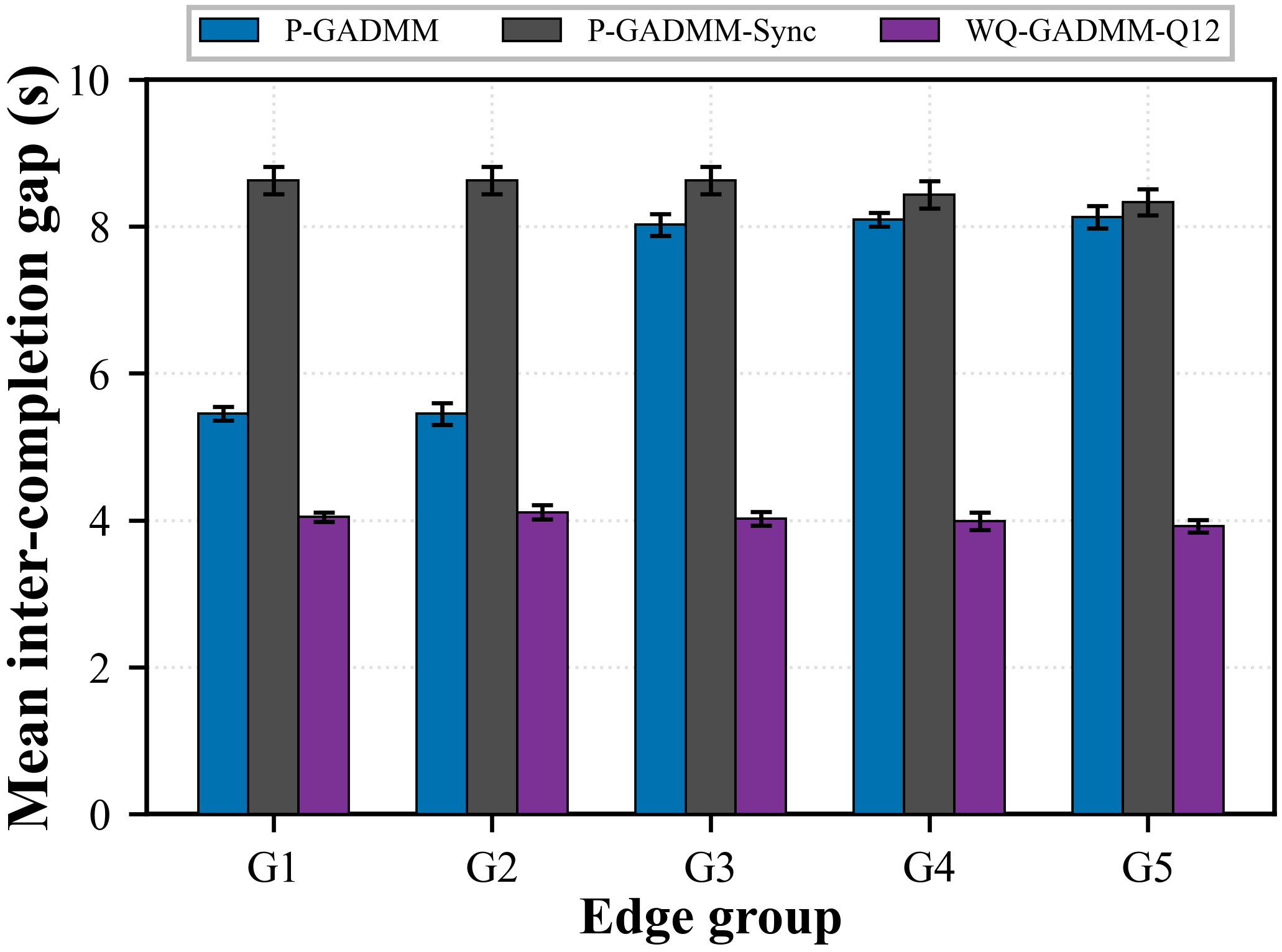}}
\hfill
\subfloat[CIFAR-10]{%
\includegraphics[width=0.49\linewidth]{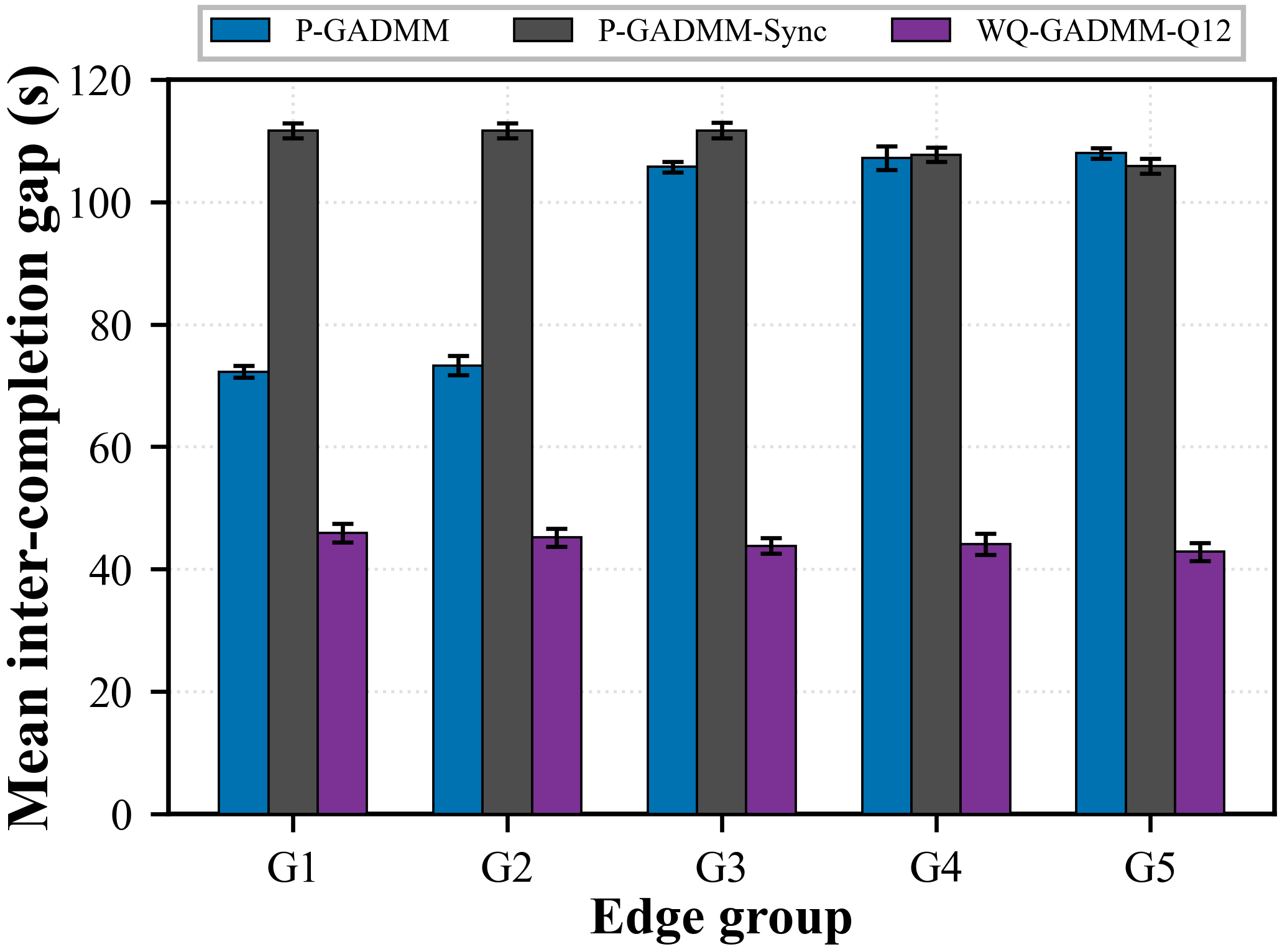}}
\caption{Mean inter-completion gap across edge groups.}
\label{fig:service_gap}
\end{figure}

\begin{table*}[t]
\centering
\caption{Accuracy and resource efficiency under communication and workload constraints.}
\label{tab:resource_efficiency}
{\footnotesize
\setlength{\tabcolsep}{8pt}
\renewcommand{\arraystretch}{1.12}
\begin{tabular}{clccc}
\toprule
& & Fixed communication budget
& \multicolumn{2}{c}{Fixed training workload} \\
\cmidrule(lr){3-3}
\cmidrule(lr){4-5}
Dataset & Method
& Accuracy (\%)
& \makecell{Communication\\ratio}
& \makecell{Simulated time\\ratio} \\
\midrule
\multirow{4}{*}{MNIST}
& P-GADMM       & $75.68 \pm 3.67$ & 1.000 & 0.802 \\
& P-GADMM-Sync  & $77.91 \pm 2.86$ & 1.000 & 1.008 \\
& WQ-GADMM-FP32 & $91.38 \pm 2.14$ & 1.000 & 1.000 \\
& WQ-GADMM-Q12  & $\mathbf{96.05 \pm 0.27}$
& \textbf{0.375} & \textbf{0.474} \\
\midrule
\multirow{4}{*}{CIFAR-10}
& P-GADMM       & $53.05 \pm 0.62$ & 1.000 & 0.821 \\
& P-GADMM-Sync  & $55.04 \pm 0.16$ & 1.000 & 0.999 \\
& WQ-GADMM-FP32 & $51.12 \pm 0.47$ & 1.000 & 1.000 \\
& WQ-GADMM-Q12  & $\mathbf{67.58 \pm 0.63}$
& \textbf{0.375} & \textbf{0.405} \\
\bottomrule
\end{tabular}

\vspace{1mm}
\parbox{0.9\textwidth}{\footnotesize
Ratios are normalized to WQ-GADMM-FP32.}
}
\end{table*}

Table~\ref{tab:fairness} shows that P-GADMM achieved Jain fairness indices above 0.96, but covered all groups in only about half of the observation intervals. P-GADMM-Sync and WQ-GADMM-Q12 both achieved fairness indices and coverage ratios of 1.000, consistent with their requirement that each group update once per window. Fig.~\ref{fig:coverage} confirms full-group coverage for both methods throughout the evaluated intervals. Their update timing nevertheless differed: Fig.~\ref{fig:service_gap} shows that WQ-GADMM-Q12 had shorter inter-completion gaps than either baseline. Compared with P-GADMM-Sync, its mean gap was reduced by 52.9\% on MNIST and 59.6\% on CIFAR-10. Thus, WQ-GADMM-Q12 maintained complete group coverage while reducing the time between successive group update completions.

\subsection{Resource Efficiency}

Finally, we compared accuracy under a common communication budget and resource consumption under a fixed training workload. The communication budgets were approximately 3.18~Gbit for MNIST and $2.801\times10^3$~Gbit for CIFAR-10. Under the fixed workload, communication volume and simulated wall-clock time were normalized to WQ-GADMM-FP32:
\[
\widetilde C_m=\frac{C_m}{C_{\mathrm{FP32}}},
\qquad
\widetilde T_{\mathrm{sim},m}
=\frac{T_{\mathrm{sim},m}}{T_{\mathrm{sim},\mathrm{FP32}}},
\]
where $C_m$ and $T_{\mathrm{sim},m}$ denote the communication volume and simulated wall-clock time of method $m$, respectively.

Table~\ref{tab:resource_efficiency} compares accuracy under a common communication budget and resource consumption under a fixed training workload. Q12 achieved higher accuracy than FP32 under the same communication budget because quantization allowed more model updates. Under the fixed workload, Q12 used 37.5\% of the FP32 communication volume, while its simulated wall-clock time was 47.4\% of FP32 on MNIST and 40.5\% on CIFAR-10. These results show that Q12 supported more training within the communication budget and reduced communication volume and simulated completion time for the same workload.

\section{Conclusion}
\label{sec:conclusion}

This paper proposed WQ-GADMM for distributed optimization in heterogeneous edge networks with limited group activation capacity. The method combines fixed grouping, windowed scheduling, and bidirectional quantization, with one update from each group before every global update. It also allows bounded model staleness and inexact local updates. Under the stated assumptions and parameter conditions, we established an average squared KKT residual bound for smooth nonconvex objectives. The bound contains an $O(1/K)$ term and a quantization-dependent error term. Experiments on MNIST and CIFAR-10 showed that WQ-GADMM-Q12 reduced communication volume and simulated training time while maintaining accuracy comparable to WQ-GADMM-FP32. It also maintained complete group coverage and achieved shorter inter-completion gaps than the evaluated baselines. Future work will address time-varying resource conditions and evaluate WQ-GADMM on physical edge platforms.

\end{document}